\documentclass[11pt]{article}

\usepackage{algorithmic}
\usepackage{amssymb,amsmath}
\usepackage{graphicx}
\usepackage{epsfig}
\usepackage{multirow}
\usepackage{color}
\usepackage{algorithm}
\usepackage[utf8]{inputenc}

\newtheorem{theorem}{Theorem}[section]

\newtheorem{lemma}[theorem]{Lemma}
\newtheorem{claim}[theorem]{Claim}

\newtheorem{definition}{Definition}[section]

 \newtheorem{observation}[theorem]{Observation}
                      {}

                      {}

\def\squarebox#1{\hbox to #1{\hfill\vbox to #1{\vfill}}}
\newcommand{\qed}{\hspace*{\fill}
\vbox{\hrule\hbox{\vrule\squarebox{.667em}\vrule}\hrule}\smallskip}

\def\eod{\vrule height 6pt width 5pt depth 0pt}
\newenvironment{proof}{\noindent {\bf Proof:} \hspace{.677em}}
                      {\hspace*{\fill}{\eod}}
\DeclareMathOperator*{\argmin}{arg\min}
\DeclareMathOperator*{\argmax}{arg\max}

\newcommand{\ceil}[1]{\left\lceil #1 \right\rceil}
\newcommand{\floor}[1]{\left\lfloor #1 \right\rfloor}

\newcommand{\R}{\mbox{I$\!$R}}
\newcommand{\N}{\mbox{I$\!$N}}

\newcommand{\remove}[1]{}

\begin{document}
\title{The Efficiency of Best-Response Dynamics
\footnote{A preliminary version appears in the proc. of the 10th International Symposium on Algorithmic Game Theory (SAGT), September 2017.}}

\author{Michal Feldman
\thanks{Tel-aviv University, Israel, Email: michal.feldman@cs.tau.ac.il. The work of the first two authors was partially supported by the European Research Council under the European Union's Seventh Framework Programme (FP7/2007-2013) / ERC grant agreement number 337122.}
\and Yuval Snappir \thanks{Tel-aviv University, Israel, {Email: yuvalsnappir@mail.tau.ac.il}}
\and Tami Tamir \thanks{ The Interdisciplinary Center, Herzliya, Israel. {Email: tami@idc.ac.il}}}
\date{}
\maketitle

\begin{abstract}
Best response (BR) dynamics is a natural method by which players proceed toward a pure Nash equilibrium via a local search method.
The quality of the equilibrium reached may depend heavily on the order by which players are chosen to perform their best response moves.
A {\em deviator rule} $S$ is a method for selecting the next deviating player.
We provide a measure for quantifying the performance of different deviator rules.
The {\em inefficiency} of a deviator rule $S$ is the maximum ratio, over all initial profiles $p$, between the social cost of the worst equilibrium reachable by $S$ from $p$ and the social cost of the best equilibrium reachable from $p$.
This inefficiency always lies between $1$ and the {\em price of anarchy}.

We study the inefficiency of various deviator rules in network formation games and job scheduling games (both are congestion games, where BR dynamics always converges to a pure NE).
For some classes of games, we compute optimal deviator rules.
Furthermore, we define and study a new class of deviator rules, called {\em local} deviator rules. Such rules choose the next deviator as a function of a restricted set of parameters, and satisfy a natural condition called {\em independence of irrelevant players}.
We present upper bounds on the inefficiency of some local deviator rules, and also show that for some classes of games, no local deviator rule can guarantee inefficiency lower than the price of anarchy.
\end{abstract}

%
%



\section{Introduction}

Nash equilibrium (NE) is perhaps the most popular solution concept in games.
It is a strategy profile from which no individual player can benefit by a unilateral deviation.
However, a Nash equilibrium is a declarative notion, not an algorithmic one.
To justify equilibrium analysis, we have to come up with a natural behavior model that leads the players of a game to a Nash equilibrium.
Otherwise, the prediction that players play an equilibrium is highly questionable.
Best response dynamics is a simple and natural method by which players proceed toward a NE via the following local search method: as long as the strategy profile is not a NE, an arbitrary player is chosen to improve her utility by deviating to her best strategy given the profile of others.

Work on BR dynamics advanced in two main avenues:
The first studies whether BR dynamics converges to a NE, if one exists, e.g., \cite{Mil96,HK12} and references therein.
The second explores how fast it takes until BR dynamics converges to a NE, e.g., \cite{EM05,FPT04,Syr10,IM+05}.

It is well known that BR dynamics does not always converge to a NE, even if one exists.
However, for the class of finite {\em potential games} \cite{Ros73,MS96}, a pure NE always exists, and BR dynamics is guaranteed to converge to one of the equilibria of the game.
A potential game is one that admits a {\em potential function} --- a function that assigns a real value to every strategy profile, and has the miraculous property that for any unilateral deviation, the change in the utility of the deviating player is mirrored accurately in the potential function.
The proof follows in a straight forward way from the definition of a potential game.
Due to the mirroring effect, and since the game is finite, the process of BR updates must terminate, and this happens at some local minimum of the potential function, which is a NE by definition.
While BR dynamics is guaranteed to converge, convergence may take an exponential number of iterations, even in a potential game. \cite{AD+08} showed that a network design game with costs obtained by a sum of construction cost and latency cost, which is a potential game, convergence via best response dynamics can be exponentially long.

Our focus in this work is different than the directions mentioned above.
The description of BR dynamics leaves the choice of the deviating player unspecified. Thus, BR dynamics is essentially a large family of dynamics, differing from one another in the choice of who would be the next player to perform her best response move.
In this paper, we study how the choice of the deviating player (henceforth a {\em deviator rule}) affects the efficiency of the equilibrium reached via BR dynamics.

Our contribution is three fold.

First, we introduce a new measure for quantifying the performance of different deviator rules.
This measure  can be used to quantify the performance of different deviator rules in various settings, beyond the ones considered in this paper.

Second, we introduce a natural class of simple deviator rules, we refer to as {\em local}.
Local deviator rules are simpler to apply, since they are based on limited information. In practice, simple deviator rules should be preferred over more complicated ones.
Our results help in quantifying the efficiency loss incurred due to simplicity.

Finally, we quantify the inefficiency of various deviator rules in two paradigmatic congestion games, namely network formation games and job scheduling games. Our results distinguish between games where local deviator rules can lead to good outcomes and games for which any local deviator rule performs poorly.

%

\subsection{Model and Problem Statement}
\label{sec:problem}

A game $G$ has a set $N$ of $n$ players. Each player $i$ has a strategy space $P_i$, from which she chooses a strategy $p_i \in P_i$.
A strategy profile is a vector of strategies for each player, $p = (p_1, \ldots, p_n)$.
The strategy profile of all players except player $i$ is denoted by $p_{-i}$, and it is convenient to denote a strategy profile $p$ as $p=(p_i, p_{-i})$.
Each player has a cost function $c_i: P \rightarrow \mathbb{R}^{\geq 0}$, where $c_i(p)$ denotes player $i$'s cost in the strategy profile $p$.
Every player wishes to minimize her cost.
There is also a social objective function, mapping each strategy profile to a social cost.


Given a strategy profile $p$, the best response of player $i$ is $BR_i(p) = \argmin_{p'_i \in P_i} c_i(p'_i,p_{-i})$; i.e., the set of strategies that minimize player $i$'s cost, fixing the strategies of all other players.
Player $i$ is said to be {\em suboptimal} in $p$ if the player can reduce her cost by a unilateral deviation, i.e., if $p_i \not\in BR_i(p)$.
If no player is suboptimal in $p$, then $p$ is a {\em Nash equilibrium} (NE) (in this paper we restrict attention to pure NE; i.e., an equilibrium in pure strategies).

Given an initial strategy profile $p^0$, a best response sequence from $p^0$ is a sequence $\langle p^0, p^1, \ldots \rangle$ in which for every $T=0,1,\ldots$ there exists a player $i \in N$ such that
$p^{T+1} = (p'_{i}, p^{T}_{-i})$, where
$p'_{i}  \in BR_i(p^{T}_{-i})$. In this paper we restrict attention to games in which every BR sequence is guaranteed to converge to a NE.

{\bf Deviator rules and their inefficiency.}
A {\em deviator rule} is a function $S : P \rightarrow N$ that given a profile $p$, chooses a deviator among all suboptimal players in $p$.
The chosen player then performs a best response move (breaking ties arbitrarily).
Given an initial strategy profile $p^0$ and a deviator rule $S$ we denote by $NE_{ S }(p^0)$ the set of NE that can be obtained as the final profile of a BR sequence $\langle p^0,p^1, \ldots\rangle$, where for every $T \geq 0$, $p^{T+1}$ is a profile resulting from a deviation of $S(p^T)$ (recall that players break ties arbitrarily, thus this is a set of possible Nash equilibria).

Given an initial profile $p^0$, let $NE(p^0)$ be the set of Nash equilibria reachable from $p^0$ via a BR sequence,
and let $p^{\star}(p^0)$ be the best NE reachable from $p^0$ via a BR sequence, that is, $p^{\star}( p^0 ) =
\argmin_{p \in NE(p^0) } SC( p )$, where $SC:P\rightarrow \mathbb{R}$ is some social cost function.

The {\em inefficiency} of a deviator rule $S$ in a game $G$, denoted $\alpha_S^G$, is defined as the worst ratio, among all initial profiles $p^0$, and all NE in $NE_S(p^0)$, between the social cost of the worst NE reachable by $S$ (from $p^0$) and the social cost of the best NE reachable from $p^0$.
I.e.,
$$\alpha_S^G = \sup_{p^0} \max_{p\in NE_S(p^0)}\frac{SC(p)}{SC(p^\star(p^0))}.$$
For a class of games $\mathcal{G}$, the inefficiency of a deviator rule $S$ with respect to $\mathcal{G}$ is defined as the worst case inefficiency over all games in $\mathcal{G}$:
$\alpha_S^{\mathcal{G}}= \sup_{G \in \mathcal{G} } \{ \alpha_S^G \}.$
A deviator rule with inefficiency $1$ is said to be {\em optimal},
 i.e., an optimal deviator rule is one that for every initial profile reaches a best equilibrium reachable from that initial profile.

The following observation, shows that the inefficiency of every deviator rule is bounded from above by the \emph{price of anarchy} (PoA) \cite{KP99,Pap01}.
Recall that the PoA is the ratio between the cost of the worst NE and the cost of the social optimum, and is used to quantify the loss incurred due to selfish behavior.

\begin{observation}
\label{ob:alpha_range}
For every game $G$ and for every deviator rule $S$ it holds that the inefficiency of $S$ is at least $1$ and bounded from above by the PoA. 
\end{observation}
\begin{proof}
\label{app:proof-poa}
For every game $G$, initial profile $p^0$, and deviator rule $S$, by definition, $p^\star(p^0)=\argmin_{p \in NE(p^0) } SC( p )$. Since $NE_S(p^0) \subseteq NE(p^0)$ we know that $SC(p^\star(p^0)) \leq SC(NE_s(p^0))$ and therefore $\frac{SC(NE_S(p^0))}{SC(p^\star(p^0))} \geq 1$. Also,
$\frac{SC(NE_S(p^0))}{SC(p^\star(p^0))} \leq \frac{SC(NE_S(p^0))}{ \min_{p \in P} SC( p ) } \leq \frac{{ \max_{p \in P} SC( p ) }}{{ \min_{p \in P} SC( p ) }} = PoA(G)$.
Since $1 \leq \frac{SC(NE_S(p^0))}{SC(p^\star(p^0))} \leq PoA(G)$, we conclude $1 \leq {\sup_{p^0}}\frac{SC(NE_S(p^0))}{SC(p^\star(p^0))} \leq PoA(G)$ and therefore $1 \leq \alpha_S^G \leq PoA(G)$.
\end{proof}

{\bf Local deviator rules.}
We define and study a class of simple deviator rules, called {\em local} deviator rules.
Local deviator rules are defined with respect to state vectors, that represent the state of the players in a particular strategy profile.
Given a profile $p$, every player $i$ is associated with a state vector $v_i$, consisting of several parameters that describe her state in $p$ and in the strategy profile obtained by her best response.
The specific parameters may vary from one application to another.
A vector profile is a vector $v=(v_1, \ldots, v_n)$, consisting of the state vectors of all players.
A deviator rule is said to be {\em local} if it satisfies the independence of irrelevant players condition, defined below.

\begin{definition}
A deviator rule $S$ satisfies {\em independence of irrelevant players} (IIP) if for every two state vectors $v_{ i_1 } ,v_{ i_2 }$, and every two vector profiles $v,v'$ such that $v=(v_{i_1}, v_{i_2}, v_{-i_1,i_2})$, and $v'=(v_{i_1}, v_{i_2}, v'_{-i_1,i_2})$ \footnote{Note that the profiles $v$ and $v'$ may correspond to different sets of players.}, if $S(v)=i_1$, then $S(v') \neq i_2$.
\end{definition}

The IIP condition means that if the deviator rule chooses a state vector $v_i$ over a state vector $v_j$ in one profile, then, whenever these two state vectors exist, the deviator rule would not choose vector $v_j$ over $v_i$.
Note that this condition should hold even across different game instances and even when the number of players is different.
Many natural deviator rules satisfy the IIP condition. For example, suppose that the state vector of a player contains her cost in the current profile and her cost in the profile obtained by her best response; then, both $(i)$ max-cost, which chooses the player with the maximum current cost, and $(ii)$ max-improvement, which chooses the player with the maximum improvement, are local deviator rules.

{\bf Congestion games.}
A congestion game has a set $E$ of $m$ resources, and the strategy space of every player $i$ is a collection of sets of resources; i.e., $P_i \subseteq 2^E$.
Every resource $e \in E$ has a cost function $f_e:\N\rightarrow \R$, where $f_e(\ell)$ is the cost of resource $e$ if $\ell$ players use resource $e$.
The cost of player $i$ in a strategy profile $p$ is
$c_i(p) = \sum_{e \in p_i} f_e(\ell_e(p))$, where $\ell_e(p)$ is the number of players that use resource $e$ in the profile $p$.
Every congestion game is a potential game \cite{MS96}, thus admits a pure NE, and moreover, every BR sequence converges to a pure NE.
In this paper we study the efficiency of deviator rules in the following congestion games:

{\em Network formation games} \cite{AD+08}: There is an underlying graph, and every player is associated with a pair of source and target nodes $s_i, t_i$. The strategy space of every player $i$ is the set of paths from $s_i$ to $t_i$. The resources are the edges of the graph, every edge $e$ is associated with some fixed cost $c_e$, which is evenly distributed by the players using it. That is, the cost of an edge $e$ in a profile $p$ is $f_e(p) = c_e / \ell(p)$.
In network formation games the cost of a resource decreases in the number of players using it.
We also consider a weighted version of network formation games on parallel edge networks, where players have weights and the cost of an edge is shared proportionally by its users.
The social cost function here is the sum of the players' costs; that is $SC( p ) = \sum_{ i \in N } c_i( p )$.

The state vector of a player in a network formation game, in a profile $p$, consists of:
\begin{enumerate}
\item Player $i$'s cost in $p$: $c_i( p )$
\item The cost of player $i$'s path: $\sum_{ e \in p_i } c_e$
\item Player $i$'s cost in the profile obtained from a best response of $i$: $c_i( p^{ \prime }( i ) )$ (where $p^{ \prime }( i ) = ( p_{ -i } , BR_i( p_{ -i } ) )$ is the profile obtained from a best response of $i$).
\item The cost of player $i$'s path in the profile $p^{ \prime }( i )$: $\sum_{ e \in BR_i( p_{ -i } ) } c_e$.
\end{enumerate}
In weighted instances, the state vector includes player $i$'s weight as well.

{\em Job scheduling games} \cite{Voc07}: The resources are machines, and players are jobs that need to be processed on one of the machines. Each job has some length, and the strategy space of every player is the set of the machines. The load on a machine in a strategy profile $p$ is the total length of the jobs assigned to it. The cost of a job is the load on its chosen machine. We also consider games with conflicting congestion effect (\cite{FT12,CG11}), where jobs have unit length and in addition to the cost associated with the load, every machine has an activation cost $B$, shared by the jobs assigned to it.
The social cost function here is the \textit{makespan}, that is $SC( p ) = \max_{ i \in N } c_i( p )$.
The state vector of a job (player) in a job scheduling game, in a profile $p$, consists of the job's length, its current
machine, and the loads on the machines.

\subsection{Our Results}
In Section \ref{sec:NFG} we present our results for network formation games.
Some of our results refer to restricted graph topologies. These topologies are defined in Appendix \ref{sec:nettopo}.
We first show that in general network formation games, finding a BRD-sequence that approximates the optimum better than the PoA is NP-hard. Therefore, we restrict attention to subclasses: We study symmetric games, where all the players share the same source and target nodes.
We show that the local {\em Min-Path} deviator rule, which chooses a player with the cheapest best response path, is optimal.
In contrast, the local {\em Max-Cost} deviator rule has the worst possible inefficiency, $n$ (which matches the PoA for this game).
We then consider asymmetric network formation games.
Unfortunately, the optimality of Min-Path does not carry over to asymmetric network formation games, even when played on series of parallel paths (SPP) networks.
In particular, the inefficiency of Min-Path in single-source multi-targets instances is $\theta(|V|)$, and for multi-sources multi-targets instances, it further grows to $\theta(2^{|V|})$.
On the positive side, we show a $poly( n , |V| )$ dynamic-programming algorithm for finding an optimal BR sequence for network formation games played on SPP networks in single-source multi-targets instances.
We also show a $poly( n , |V| )$ such algorithm for multi-sources multi-targets instances that admit ``proper intervals" (see Section \ref{app:DPproper} for formal definitions).
For network formation games played on extension-parallel networks we show that every local deviator rule has an inefficiency of $\Omega(n)$.

In Section \ref{sec:weightedSymm} we study network formation games with weighted players.
It turns out that weighted players lead to quite negative results.
We show that even in the simplest case of parallel-edge networks, it is NP-hard to find an optimal BR-sequence, and no local deviator rule can ensure a constant inefficiency.
Moreover, the Min-Path deviator rule has inefficiency $\Omega(n)$, even in symmetric games on SPP networks, and even if the ratio between the maximal and minimal weights approaches $1$.
This result is quite surprising in light of the fact that Min-Path is optimal (i.e., has inefficiency $1$) for any symmetric network formation game with unweighted players.

Job scheduling games with $m$ identical machines are studied in Section \ref{sec:sched}.
A job's ($=$ player's) state vector in job scheduling games includes the job's length, its machine, and the machines' loads.
Thus, local deviator rules capture many natural rules, such as Longest-Job, Max-Cost, Max-Improvement, and more.
We show that no local deviator rule can guarantee inefficiency better than the $PoA$ in this class of games (which is $\frac{ 2m }{ m + 1 }$).
%
In contrast, for job scheduling games with conflicting congestion effects \cite{FT12} we present an optimal local deviator rule.

Positive results on local deviator rules imply that a centralized authority that can control the order of deviations can lead the population to a good outcome, by considering merely local information captured in the close neighborhood of the current state.
In contrast, negative results for local deviator rules imply that even if a centralized authority can control the order of deviations, in order to converge to a good outcome, it cannot rely only on local information; rather, it must be able to perform complex calculations and to consider a large search space.

\subsection{Related work}
A lot of research has been conducted on the analysis of congestion games using a game-theoretic approach.
The questions that are commonly analyzed are Nash equilibrium existence,
the convergence of BR dynamics to a NE, and the loss incurred due to selfish behavior -- commonly quantified according to the \emph{price of anarchy} \cite{KP99,Pap01} and price of stability (PoS) \cite{AD+08} measures.


Our work addresses mainly congestion games. Specifically, we consider network formation games and job scheduling games and variants thereof. It is well known that every congestion game is a potential game \cite{Ros73,MS96} and therefore admits a PNE. In potential games every BR dynamics converges to a PNE. However, the convergence time may, in general, be exponentially long. It has been shown in \cite{SB16} that in random potential games with $n$ players in which every player has at most $a$ strategies, the worst case convergence time is $n \cdot a^{ n - 1 }$. The average case convergence time in random potential games is sublinear but along with the computation complexity of finding the best response has running time of $\mathcal{O}( a \cdot n )$. Results regarding the convergence time in network formation games are shown in \cite{AD+08}, \cite{FPT04}. They showed that in general, finding a PNE is PLS-complete. In the restricted case of symmetric players, in a game with both positive and negative congestion effects, the convergence is polynomial. Analysis of job scheduling games is presented in \cite{EA+03}, where it is shown that it may take exponential number of steps to converge in general cases. 

The observation that convergence via best response dynamics can be exponentially long has led to a large amount of work aiming to identify special classes of congestion games, where BRD converges to a Nash equilibrium in polynomial time or even linear time, as shown by \cite{AD+08} for games with positive congestion effects and by \\ \cite{EA+03} for games with negative congestion effects. This agenda has been the focus of \cite{F10} that considered symmetric network formation games with negative congestion effect played on an extension-parallel graph, and showed that the convergence is bounded by $n$ steps. For resource selection games (i.e., where feasible strategies are composed of singletons), it has been shown in \cite{IM+05} that better-response dynamics converges within at most $mn^2$ steps for general cost functions (where $m$ is the number of resources).

Another direction of research considers an {\em approximate Nash equilibrium}, also known as {\em $\epsilon$-Nash equilibrium}. An approximate-NE is a strategy profile in which no player can significantly improve her utility by a unilateral deviation. \cite{CF+11} presented an algorithm that identifies a polynomially long sequence of BR moves that leads to an approximate NE in congestion games with linear latency functions. \cite{B+15} showed for the more general case of polynomially decreasing resources' cost functions, given an algorithm to find $\epsilon$-best response, convergence time is polynomial. In weighted network formation games which are known to not always have a PNE, \cite{CR09} study the relations between the $\epsilon$ needed to ensure existence of an $\epsilon$-PNE and the inefficiency of the solution. \cite{CS11} showed that it takes $poly( n , \frac{ 1 }{ \epsilon } )$ steps to converge to a PNE in symmetric network formation games, and that even when using different thresholds $\epsilon_{ i }$ the amount of deviations conducted by player $i$ will be $poly( n , \frac{ 1 }{ \epsilon_{ i } } )$.




Another technique for controlling the convergence time of BRD is using a specific deviator rule. \cite{EA+03} compared the convergence time of different deviator rules for different variants of job scheduling games. For example, for instances of identical machines they show that the Max-Weight-Job deviator rule ensures convergence in at most $n$ steps while convergence by the Min-Weight-Job can be exponential. They also show that the FIFO deviator rule converges within time $\mathcal{O}( n^2 )$ and the expected convergence time of a random order is also $\mathcal{O}( n^2 )$.
\cite{FT12} and \cite{CG11} considered a variant of job scheduling with both negative and positive congestion effects and analyzed their PoA and PoS. We will refer to that model in this work. \cite{FT15} showed that using the Max-Cost deviator rule, tightly bounds the convergence rate to $ 1.5 n $ (worst-case) in these games. The Max-Cost deviator rule was also considered in \cite{KL13} for {\em Swap-Games} \cite{A+10}. In Swap-Games with initial strategies corresponding to a tree it has been shown that the convergence time is $\mathcal{O}( n^{ 3 } )$, but using the Max-Cost deviator rule improves the bound to $\mathcal{O}( n )$.

Although every congestion game have a PNE and BRD always converges to one, examples for variants of congestion games that do not have a PNE
or that BRD does not necessarily converge to one are shown in the literature. In particular, \cite{CR09} presents a network-formation game with weighted players that does not have a PNE; \\ \cite{HK12} presents a characterization of the resources' cost functions that ensures existence of PNE, and proves various conditions for the {\em finite-improvement-property} (FIP) that implies that every best response dynamic converges to a PNE. \cite{Mil96} shows that in a variant of congestion games in which every player has a specific decreasing cost function for every resource, the FIP condition does not necessarily hold. Though there is always a PNE and best response dynamics can always converges to it, it is shown that there can be infinitely long best response sequences and conditions for them to occur are presented. \cite{AKT14} presents network formation games with a variation that players do not aim to connect a source and a sink, but to select a path corresponding to a word in a regular language, assuming the edges are labeled by an alphabet. They showed cases with no PNE and moreover, that computing a player's best response is NP-hard, as was shown for another variant of network formation games presented in \cite{FLMPS03}. A variant of job scheduling game on unrelated parallel machines, studied in \cite{AT16} is shown to have a PNE only for unit-cost machines. All the above works imply that the existence of a PNE, as well as the convergence of BRD to one are not trivial. Moreover, even when a PNE exists and BRD is guaranteed to converge, the implementation of best response dynamics can be computationally hard.

BRD has been studied also in games that do not converge to a PNE. The notion of {\em dynamic inefficiency} was defined in \cite{BF+11} as the average of social costs in a best response infinite sequence (for games that do not possess the finite improvement property). They considered the effect of the chosen player on the obtained efficiency in games where BR dynamics can cycle indefinitely. They consider job scheduling games, hotelling model and facility location games and study the dynamic inefficiency of {\em Random Walk}, {\em Round Robin} and {\em Best Improvement} deviator rules.



\section{Network Formation Games}
\label{sec:NFG}
%

In this section we study network formation games.
We consider two natural local deviator rules, namely {\em Max-Cost} and {\em Min-Path}.
The Max-Cost deviator rule chooses a suboptimal player that incurs the highest cost in the current strategy profile $p$, i.e.,
$Max-Cost( p ) \in \argmax_{ \{ i \in N | p_i \not\in BR_i(p) \} } c_i( p )$.
The Min-Path deviator rule chooses a suboptimal player whose path in the profile obtained from a best response move is cheapest, i.e.,
$Min-Path( p ) \in \argmin_{ \{ i \in N | p_i \not\in BR_i(p) \} } \sum_{ e \in BR_i( p ) } c_e$. Both rules are local deviator rules.

Recall that some of our results refer to restricted graph topologies, defined in Appendix \ref{sec:nettopo}.
This section is organized as follows.
Section \ref{sec:symNFG} includes our analysis of symmetric games. In Sections \ref{sec:spp} and \ref{sec:ep}
we study games played on series of parallel paths networks and extension parallel networks, respectively.
Finally, in Section \ref{sec:weightedSymm} we study weighted network formation games.

\subsection{ The BRD-Inefficiency in General Network Formation Games }
\label{sec:NFGhardness}

We start with the most general form of network formation games, with no restriction on the network topology or the players' objectives. We prove that for the general case it is NP-hard to find a deviator rule whose inefficiency is lower than the PoA.

\begin{theorem}
In network formation games, it is NP-hard to find a deviator rule with inefficiency lower than $\Omega(n)$.
\end{theorem}
\begin{proof}
Given a game, an initial strategy profile, and a value $k$, the associated decision problem is whether there exists an order according to which the players are activated and results in a NE whose social cost is at most $k$. We show a reduction from the {\em Partition} problem: Given a set of numbers $ \{ a_1 , a_2 , ... , a_n \} $ such that $\sum_{ i \in [n] } a_i = 2$, the goal is to find a subset $ I \subseteq [n] $ such that $\sum_{ i \in I } a_i = \sum_{ i \in [n] \backslash I } a_i = 1$. This problem is NP-hard even if it is known that such a subset $I$ exists. Given an instance of {\em Partition}, such that for some $ I \subseteq [n] $, it holds that $\sum_{ i \in I } a_i =1$, consider the network depicted in Figure \ref{fig:NFGhard}, with the following initial strategy profile of $2n+2=O(n)$ players:
\begin{itemize}
\item Players $ 1 ,\ldots, n $: Each Player $i \in [n]$, has a source $v_{ i - 1 }$ and a target $ v_{ i } $ and she initially uses the upper edge of cost $3$ by herself. The alternative edge for Player $i$ costs $a_i$.
\item Players $ n+1 ,\ldots, 2n$: $n$ players whose objective is an $\langle s', t' \rangle$-path. Initially, they all use the lower $(s', t')$-edge of cost $n$, thus, each of them pays $1$. 
\item Player $a$ whose objective is an $\langle s_a, t_{a,b}\rangle$-path. She initially uses the upper path and share its second edge with Player $b$, so her initial cost is $1.5 + \frac{ \epsilon }{ 2 }$.
\item Player $b$ whose objective is an $\langle s_b, t_{a,b}\rangle$-path. She initially uses the upper path and share its second edge with Player $a$, so her initial cost is $1.5 + \frac{ \epsilon }{ 2 }$.
\end{itemize}

\begin{figure}[h]
\begin{center}
\includegraphics[scale=.4]{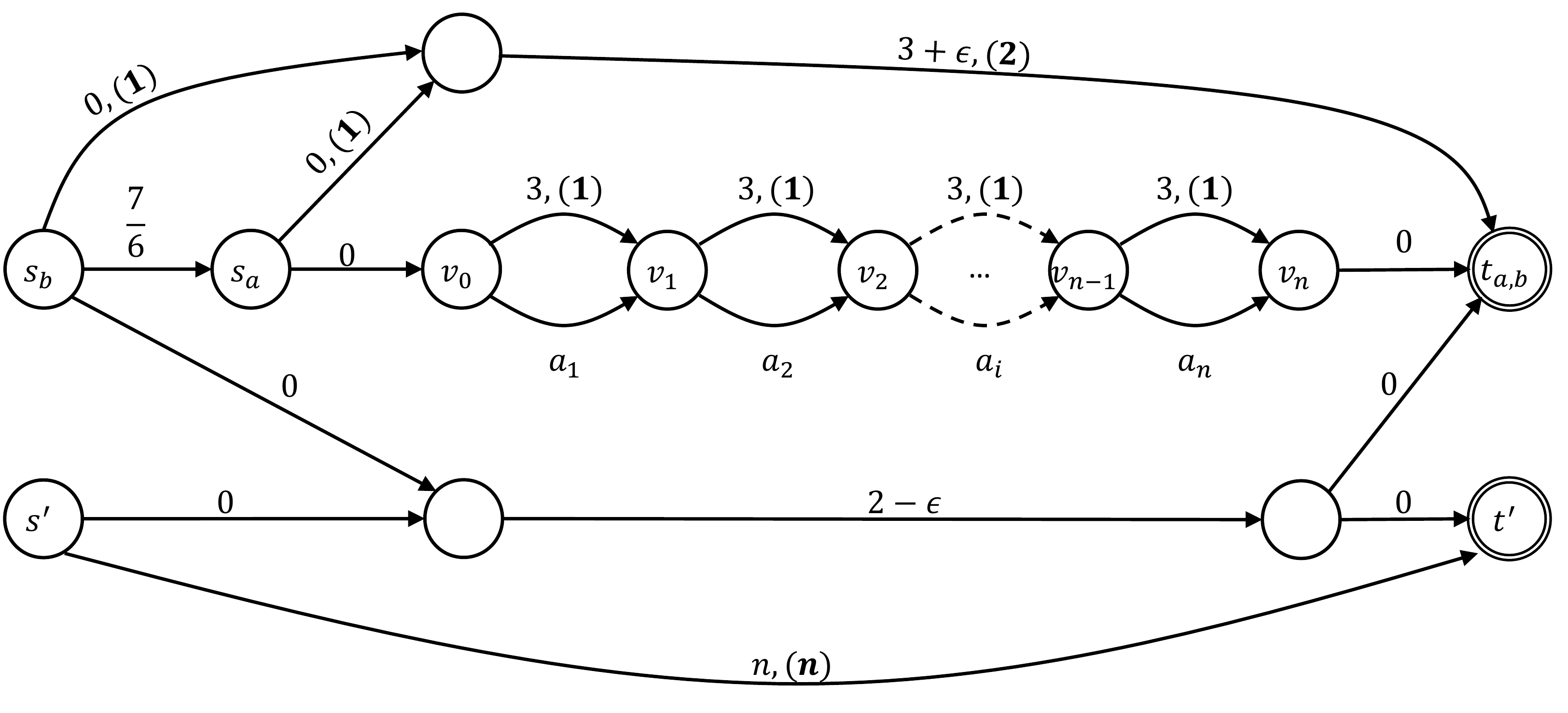}
\caption{The network constructed for a given {\em Partition} instance. Every edge is labeled by its cost and (in brackets) the number of players using it in the initial strategy profile.}
\label{fig:NFGhard}
\end{center}
\end{figure}

The following simple observations limit the possible BRD sequences of the above game instance:
\begin{enumerate}
\item Since $a_i <3$, for all $i \in [n]$, each of the first $n$ players uses the $a_i$-edge in any NE.
\item The $n$ $\langle s', t' \rangle$-players would benefit from a deviation only after the edge of cost $2 - \epsilon$ is utilized by some other player.
\item Players $a$ and $b$ will never use edges of cost $3$. If they select a path that passes through $v_0$, it will consist of all the $a_i$-edges. The initial cost of this alternative path for Player $a$ is $\sum a_i =2$. She will deviate to this path only after her total share in the $a_i$-edges is at most $1.5$.
\item As long as Player $a$ is using the upper path, deviating to a path through $v_0$ is not beneficial for Player $b$, since such a path will incur her a cost of at least $\frac{ 7 }{ 6 } + 1 $. Also, deviating to the lower path will incur her a cost of $2 - \epsilon$, which is more than her current cost. Therefore, Player $b$ will not deviate before Player $a$.
\end{enumerate}

Using the above observations we prove the statement of the reduction.
\begin{claim}
Let $I$ be a subset of $[n]$ such that $\sum_{ i \in I } a_i=1$. A deviator rule that knows $I$ can determine a BR-sequence that ends up with a NE whose social cost is $\mathcal{O}( 1 )$. If $I$ is not detected then the final NE will have social cost $\Omega(n)$.
\end{claim}

\begin{proof}
Assume that a subset $ I \subseteq [n] $ such that $\sum_{ i \in I } a_i=1$ is known. Consider a BR-sequence in which the players in $I$ deviate first (in an arbitrary order), then Player $a$ deviates, then Player $b$, and then all the suboptimal players (in an arbitrary order). Since $\sum_{ i \in I } a_i =1$, the players in $I$ utilize $a_i$-edges of total cost $1$.
When Player $a$ gets to deviate, she selects a path that passes through $v_0$. Her share in the sub-path consisting of the $a_i$-edges is $0.5$ for the utilized edges, and $1$ for the non-utilized edges. Since her current cost is $1.5 + \frac{\epsilon}2$, this deviation is indeed beneficial.

When Player $b$ gets to deviate, her share in the sub-path consisting of the $a_i$-edges would be $\frac{ 1 }{ 3 } + \frac{ 1 }{ 2 }= \frac 5 6$. Her best response is therefore the lower path since $2 - \epsilon < \frac{ 7 }{ 6 } + \frac{ 5 }{ 6 } = 2$.

After Player $b$ deviates all the players using the expensive $(s',t')$-edge will join Player $b$, and the rest of the first $n$ players will deviate to their corresponding $a_i$-edge. The social cost of the resulting NE would be $2-\epsilon + \frac{7}{6}+2 = 5 \frac{ 5 }{ 6 } - \epsilon = \mathcal{O}( 1 )$.

For the other direction of the reduction assume that the social cost of some NE reachable from the given initial profile is a constant. We show that the deviator rule must activate a subset of the players $1,\ldots,n$, corresponding to partition elements of total size $1$ before Player $a$ deviates. In order to end up with a constant social cost,  the expensive lower $(s',t')$-edge was abandoned. By the above observations, the players using this edge deviated to the $(2-\epsilon)$-edge, and in order for that to happen, Player $b$ must deviate to this $(2 - \epsilon)$-edge first. As stated, Player $b$ will not deviate before Player $a$. Also, players $1,\ldots,n$ must deviate sometime along the dynamics. Let $x_a$ ($x_b$) be the total cost of $a_i$-edges utilized by Players $1, 2,\ldots, n$ before Player $a$ ($b$) deviated. Since Player $a$ must deviate before Player $b$ it holds that $x_a \leq x_b$.

In order for the deviation to be profitable for Player $a$, it must be that
$\frac{ 1 }{ 2 } \cdot x_a + ( 2 - x_a )< 1.5 + \frac{\epsilon}2$, since her newly incurred cost consists of her share in the cost of edges she share with players that already deviated and the cost of edges she uses alone. We conclude that $x_a > 1 - \epsilon$.

In order for Player $b$ to deviate to the $(2 - \epsilon)$-edge it must be more beneficial than deviating to a path through $v_0$. The latter would consists of the edge of cost $\frac{ 7 }{ 6 }$, edges she share with Player $a$ and players in $[n]$ that already deviated, and edges that she share only with Player $a$. Formally,
$2 - \epsilon \leq \frac{ 7 }{ 6 } + \frac{ 1 }{ 3 } \cdot x_{ b } + \frac{ 1 }{ 2 } \cdot ( 2 - x_{ b } )$.
That is, $x_b \leq 1 + 6\epsilon$.

Combining the above observations, we get that $1 - \epsilon < x_a \leq x_b \leq 1 + 6 \epsilon$.
This implies that for $\epsilon < \frac 1 6 \min_i a_i$, we have $x_a=x_b=1$. The set of players that utilize the edges of total cost $1$ correspond to a set $I \subset [n]$ of sum $1$, which induces a partition.

Finally, note that any other BR-sequence leaves the lower expensive edge activated, and therefore, has total cost $\Omega(n)$.
\end{proof}
\end{proof}


\subsection{ Symmetric Network Formation Games }
\label{sec:symNFG}

A network-formation game is {\em symmetric} if all the players have the same source and target nodes.
Recall that the inefficiency of any deviator rule is upper bounded by the price of anarchy of the game.
It is well known that the PoA of network formation games is $n$ (i.e., the number of players).

We first show that the Max-Cost rule may perform as poorly as the PoA, even in symmetric games on parallel-edge networks.

\begin{observation}
\label{ob:MC}
The inefficiency of Max-Cost in symmetric network formation games on parallel-edge networks is $n$.
\end{observation}

\begin{proof}
Given $n \in \mathbb{N}$, we present an instance $G$ over $n$ players for which $\alpha_{Max-Cost}^G= n$. Consider the network depicted in Figure \ref{fig:MCBRD}.
\begin{figure}[ht]
\begin{center}
\includegraphics[height=2.25cm]{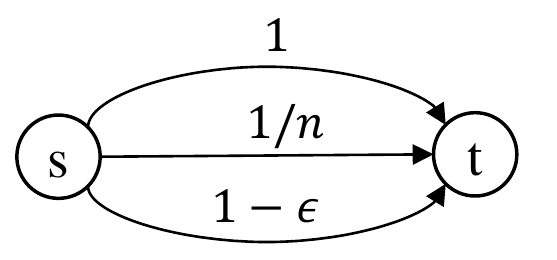}
\caption{A network on which Max-Cost has inefficiency $n$.}
\label{fig:MCBRD}
\end{center}
\end{figure}
The initial profile consists of a single player using the top edge and $n-1$ players using the bottom edge. For the first step, the cost of the player on the top edge is $1$ and the costs of the players on the bottom one is $\frac{1 - \epsilon}{n-1}$. Therefore, Max-Cost chooses the player on the upper edge. That player migrates to the bottom edge, resulting in social cost $1-\epsilon$.
However, if the players on the bottom edge deviate first (to the middle edge), then the player on the top edge will also deviate to the middle edge, resulting in a NE profile with a social cost of $\frac{1}{n}$.
The ratio converges to $n$ as $\epsilon \rightarrow 0$.
\end{proof}

On the other hand, we show that Min-Path is an optimal deviator rule, i.e., it always reaches the best NE reachable from any initial profile.
Our analysis of Min-Path is based on the following Lemma:
\begin{lemma}
\label{lem:follow}
In symmetric network formation games, the path chosen by the first deviator is the unique path that will be chosen by all subsequent players, regardless of the order in which they deviate.
\end{lemma}

\begin{proof}
Assume that player $i_1$ is the first to migrate and it chooses $p_1$ as a best response. Since $p_1$ is $i_1$'s best-response, after her deviation she cannot be suboptimal, and so are all the other players who use $p_1$ by symmetry. Assume, by a way of contradiction, that $p_2 \neq p_1$ is a best response of a suboptimal player $i_2$ not using $p_1$. Since $i_2$ can deviate to $p_1$ and incur a cost lower than the cost $i_1$ pays, but prefers to deviate to $p_2$ - a deviation of $i_1$ to $p_2$ will also be beneficial, a contradiction.
\end{proof}

Lemma~\ref{lem:follow} directly implies the optimality of Min-Path:
\begin{theorem}
\label{thm:minopt}
Min-Path is an optimal deviator rule for symmetric network formation games.
\end{theorem}

\begin{proof}
By lemma \ref{lem:follow} after the first deviation, the first deviation dictates the NE to be reached. Thus, the set of BR paths in $p^0$ are the set of reachable NE, and choosing the cheapest one among them is optimal.
\end{proof}


\subsection{Series of Parallel Paths (SPP) Networks}
\label{sec:spp}
In this section we study network formation games played on SPP networks.
An SPP network consists of $m$ {\em segments}, where each segment is a parallel-edge network.
Let $\{u_0,\ldots,u_m\}$ denote the vertex set, and for every $j \leq m$, let $E_{j}$ denote the set of edges in segment $j$ (i.e., the parallel edges connecting $u_{j-1}$ and $u_{j}$). For a player $i$, let $E(i)=\cup_{s_i < k \le t_i} E_{k}$, denote the set of edges player $i$ may choose.

Note that in an SPP network, a player's choice of an edge in $E_{j}$ is independent of any other segment in her path.
This implies that a network formation game on an SPP network consists of a sequence of symmetric games, where the set of players participating in each game varies. Combining this observation with Lemma \ref{lem:follow} implies the following.

\begin{lemma}
\label{lem:1st}
In every network formation game played on an SPP network with $m$ segments, for every $1 \le j \le m$, and every BR sequence, let $i$ be the first player in the sequence such that $E_j \in E(i)$, and let $e$ be the edge in $E_j$ chosen by player $i$.
Then $e$ is the unique edge in $E_j$ players deviate to.
\end{lemma}

\subsubsection{Optimal BR sequence for single-source, multi-targets games }
We first consider single-source, multi-targets instances.
For this case, we devise a dynamic programming algorithm that computes an optimal BR sequence.
\begin{theorem}
For any single-source multi-targets network formation game played on an SPP network, an optimal BR sequence can be computed in time $O(nm)$, where $m$ is the number of segments and $n$ is the number of players.
\end{theorem}

\begin{proof}
\label{thm:dpsinglesource}
Lemma \ref{lem:1st} applied to a game with a single source implies that a migration of a player $i$ with target $u_j$ determines the path from the source ($s=u_0$) to $u_j$ that will be used in the NE reached.

Let $N_{j}$ denote the subnetwork consisting of the {\em last} $j$ segments in the input SPP. Formally, $N_{j}=\cup_{k=m-j+1}^m E_{k}$.
Our dynamic programming solution is based on calculating an optimal solution for every suffix $N_j$ of the network. In particular, a solution for $N_m$ is a solution for the input SPP.

Let $OPT_j$ be the minimum cost of a path that can be reached in the game induced by $N_j$, and let $P_j$ denote the first player that has to migrate in order to reach it.
The base case of the dynamic programming is $j=1$. Clearly, $OPT_1$ and $P_1$ can be easily found since the game induced by $N_1$ is symmetric:
\begin{displaymath}
OPT_1 = \min_{\{i| t_i = u_{m} \}} c( BR_i( p^0 ) \cap N_1 )
\end{displaymath}
\begin{displaymath}
P_1 = \argmin_{\{i| t_i = u_{m} \}} c( BR_i( p^0 ) \cap N_1 ).
\end{displaymath}

For $j = 2,\ldots, m$ we calculate $OPT_j$ and $P_j$ as follows: For every player $i$ such that $t_i \in \{u_{m-j +1},\ldots u_m\}$, denote by $OPT_j^i$ the minimal cost of a path that can be reached in the game induced by $N_j$, if player $i$ is the first to perform her BR.
Assume $t_i=u_k$. By Lemma \ref{lem:1st}, the NE corresponding to $OPT_j^i$ consists of a prefix $e_{m-j},\ldots,e_k$ in $E_{m-j+1} \cup ... \cup E_k$, which is player $i$'s best response path. That prefix is followed by a suffix of cost $OPT_{m-k}$. Denote the prefix's cost by
$
C_{m-j,t_i}^i = c( BR_{ i }( p^0 ) \cap N_j )
$

We can compute $OPT_{ j }, P_{ j }$ using the following formulas:
\begin{displaymath}
OPT_j^i = C_{m-j,t_i}^i + OPT_{m-k} ~~ , t_i = u_k
\end{displaymath}
\begin{displaymath}
OPT_j = \min_{\{i| t_i \in \{u_{m-j},\ldots u_m\}\}}OPT_j^i
\end{displaymath}
\begin{displaymath}
P_j = \argmin_{\{i| t_i \in \{u_{m-j},\ldots u_m\}\}}OPT_j^i
\end{displaymath}

The algorithm consists of an $O(nm)$ preprocessing in which the values $C_{m-j,t_i}^i$ are calculated for all $j$ and $i$. Given $C$, it is possible to calculate each of $OPT_2,\ldots,OPT_m$ in time $O(n)$, giving a total of $O(nm)$ for the whole algorithm.
Note that the player $P_j$ is a one determining $OPT_j$, that is, $OPT_j = OPT_j^{P_j}$.
Given $OPT_j, P_j$ for all $j$, the BR sequence begins with $P_m$ as the first player, if its target is $u_k$ then the next player to perform BR will be $P_k$, etc.
\end{proof}

\subsubsection{Optimal BR sequence for multi-sources, multi-targets games with proper intervals}
\label{app:DPproper}
In this section we consider multi-sources, multi-targets instances with {\em proper intervals}.
A network formation game played on an SPP network is said to have proper intervals if for every two players $i_1,i_2$ it holds that if $s_{i_1} < s_{i_2}$ then $t_{i_1} \leq t_{i_2}$.
We denote by $[s_i, t_i]$ the set of segments from $i$'s source to $i$'s target, and call it the {\em interval} of player $i$.
I.e., $[s_i, t_i] = E_{ s_i + 1 } \cup E_{ s_i + 2 } \cup ... \cup E_{ t_i }$.

%
%
\begin{theorem}
For any multi-sources multi-targets network formation game played on an SPP network with proper intervals, an optimal BR sequence can be computed in time $O(nm^2)$, where $m$ is the number of segments and $n$ is the number of players.
\end{theorem}
As in the single-source case in the previous subsection, every deviation dictates the specific path that will be used in the NE reached. Therefore, the best NE path reached by a BR sequence that starts by a deviation of a specific player, consists of the player's BR path and the best solution in the subnetworks that the player doesn't use. Our dynamic programming solution uses this observation to find the player that the NE reached by a BR sequence that starts with her deviation is the optimal. To find that NE, we consider the player's BR path and the optimal paths on the residual, strictly shorter subnetworks.
Denote by $N_{s,t}$ the subnetwork consisting of the segments connecting $u_s$ and $u_t$. That is, $N_{s,t}=\cup_{k=s+1}^{t} E_{k}$.
The dynamic programming solution is based on calculating, for every $0 \le s \le t \le m$, an optimal solution for the subnetwork $N_{s,t}$ by iterating over the subnetworks in ascending lengths. In particular, a solution for $N_{0,m}$ is a solution for the input SPP.
Let $OPT_{s,t}$ be the minimum social cost in a NE that can be reached in the game induced by $N_{s,t}$. Let $P_{s,t}$ denote the first player that has to migrate in order to reach that optimal NE. The base cases of the dynamic programming are all the subnetworks induced by a single segment, i.e., $N_{j-1,j}$ for all $1 \leq j \leq m$. Clearly, $OPT_{j-1,j}$ and $P_{j-1,j}$ can be easily found since the game induced by $N_{j-1,j}$ is symmetric. So we initialize
\begin{displaymath}
OPT_{ j - 1, j } = \min_{ \{ i | E_j \in [s_i , t_i] \} } c( BR_i( p^0 ) \cap N_{ j - 1 , j } )
\end{displaymath}
\begin{displaymath}
P_{ j - 1, j } = \argmin_{ \{ i | E_j \in [s_i , t_i] \} } c( BR_i( p^0 ) \cap N_{ j - 1 , j } )
\end{displaymath}
In addition, for all $0 \leq j \leq m$, we initialize $OPT_{j,j}=0$.

For $t > s+1$, we compute $OPT_{s,t}$ and $P_{s,t}$ as follows: For every player $i$ for which $[s,t] \cap [s_i,t_i] \neq \emptyset$, denote by $OPT_{s,t}^i$ the minimal cost of a path that can be reached in the game induced by $N_{s,t}$, if player $i$ is the first to perform her BR. Denote by $C_{s,t}^i$ the cost of $i$'s BR on the subnetwork $N_{s,t} \cap [ s_i , t_i ]$.
\begin{displaymath}
C_{s,t}^i = c( BR_{ i }( p^0 ) \cap N_{ s , t } )
\end{displaymath}
When calculating $OPT_{s,t}^i$ We distinguish between four cases:
\begin{itemize}
\item If $N_{s,t} \subset [ s_i , t_i ]$ then $OPT_{s,t}^i = C^i_{s,t}$, i.e., if the subnetwork is contained in the $i$'s interval then its deviation will set the cost of the whole subnetwork.
\item If $E_{s+1} \in [s_i, t_i]$ then $OPT_{s,t}^i = C^i_{s,t_i} + OPT_{t_i,t}$, i.e., if the subnetwork has a suffix not contained in $i$'s interval then its deviation will set the cost of their intersection and we need to add the optimal cost of the suffix to the total cost.
\item If $E_t \in [s_i, t_i]$ then $OPT_{s,t}^i = OPT_{s,s_i} + C^i_{s_i,t}$, residue prefix equivalent to the suffix in the previous condition
\item Otherwise, $OPT_{s,t}^i = OPT_{s,s_i} + C^i_{s_i, t_i} + OPT_{t_i,t}$, combines residue prefix and suffix of the subnetwork over the player's interval. Note that since we assume a proper interval instance, no player that plays in $N_{s,s_i}$ can play in $N_{t_i,t}$ and vice versa, so the process of computing each of them is completely independent.
\end{itemize}
In each of the cases, the calculation of $OPT_{s,t}^i$ requires as sub-problems only values of $OPT_{s^{\prime},t^{\prime}}$ for which $N_{s^{\prime},t^{\prime}}$ is strictly shorter than $N_{s,t}$. Also, our base cases include $N_{j,j}$ and $N_{j-1,j}$ for every $j$, thus calculating $OPT_{s,t}^i$ can be done by iterating through $N_{s,t}$ in increasing length.
We get that
\begin{displaymath}
OPT_{s,t} = \min_{\{i| [s,t] \cap [s_i,t_i] \neq \emptyset \}}OPT_{s,t}^i,
\end{displaymath}
\begin{displaymath}
P_{s,t} = \argmin_{\{i| [s,t] \cap [s_i,t_i] \neq \emptyset \}}OPT_{s,t}^i
\end{displaymath}

The algorithm consists of an $O(n \cdot m ^ 2)$-time preprocessing in which the values $C_{s,t}^i$ are calculated for all $i,s$ and $t$. Given $C$, it is possible to calculate each of $OPT_{s,t}$ in time $O(n \cdot m ^ 2)$, giving a total of $O(n \cdot m ^ 2)$ for the whole algorithm. The set of players $P_{ s , t }$ that are calculated in the process form the optimal BR-sequence.


\subsubsection{The performance of Min-Path in SPP networks}
\label{app:SPPminPath}
After presenting non-local optimal deviator rules, we turn to analyze the inefficiency of Min-Path that was shown to be optimal for symmetric network formation games.
Given an SPP network and a BR-sequence, we say that a segment is {\em unresolved} if there are at least two players whose interval include the segment, and each of them will select a different edge in the segment if chosen to perform a BR next. The other segments are denoted {\em resolved}.
By Lemma \ref{lem:1st}, after a player performs her best response, all the segments in her interval are resolved. Thus, no player migrates more than once.
By definition, in a resolved segment, the edge chosen in every reachable Nash equilibria is determined already, therefore, migrations of players who use only resolved segments do not influence the reachable Nash equilibria and in the following analyses we ignore them. In other words,
all the deviations we consider resolve at least one segment. We denote by $R_i$ the
resolved segments after $i$ such deviations.
Let $OPT$ denote the minimal cost of a NE reachable from $p^0$ by some BR sequence. Formally, $OPT= SC(p^{\star}(p_0))$.

\begin{lemma}
\label{lem:SPP}
For any BR sequence of an SPP network instance, as long as there are unresolved segments, there exists a suboptimal player whose interval includes unresolved segments, and if this player is chosen next, then the cost of the unresolved segments she would set ("resolve") is at most OPT.
\end{lemma}
\begin{proof}
Let $p$ be an intermediate strategy profile in the BR sequence. Consider the players according to the order they deviate in some optimal BR sequence. Let $i^{ \prime }$ be the first player in this order who is suboptimal in $p$. Since no player prior to $i^{ \prime }$ in the optimal sequence is suboptimal, the segments that $i^{ \prime }$ would resolve by a deviation from $p$ are a subset of the segments she resolves in the optimal sequence. In the optimal sequence she obviously resolves these segments such that the selected edges are of total cost at most $OPT$, and therefore this is an upper bound on the total cost of unresolved segments she would set by deviating from $p$.
\end{proof}

Using the above lemma, we provide tight analysis on the performance of Min-Path for SPPs with multi-targets and single or multiple sources.
Note that in a single-source instance, every player resolves the prefix of the network corresponding to her interval.

\begin{theorem}
The inefficiency of Min-Path in SPP network formation games with single-source and multi-targets is $\theta(m)$.
\end{theorem}
\begin{proof}
\label{lem:playerBoundOptSingle}
We show that the total cost determined for the segments resolved in every iteration is at most $OPT$. Since at least one segment is resolved in each iteration, the whole network's cost is bounded by $m \cdot OPT$.
Let $i$ be the $i$-th player chosen to deviate by Min-Path and assume $i$ has unresolved segments. Let $i^{ \prime }$ be the player guaranteed by Lemma \ref{lem:SPP}. It may be that $i=i'$. Both players have the same BR path in the resolved segments and therefore differ only in their unresolved segments. Since Min-Path chose $i$, the cost of her unresolved segments is at most the cost of $i'$'s unresolved segments, which is at most OPT by Lemma \ref{lem:SPP}.

\begin{figure}[ht]
\begin{center}
\includegraphics[height=2.75cm]{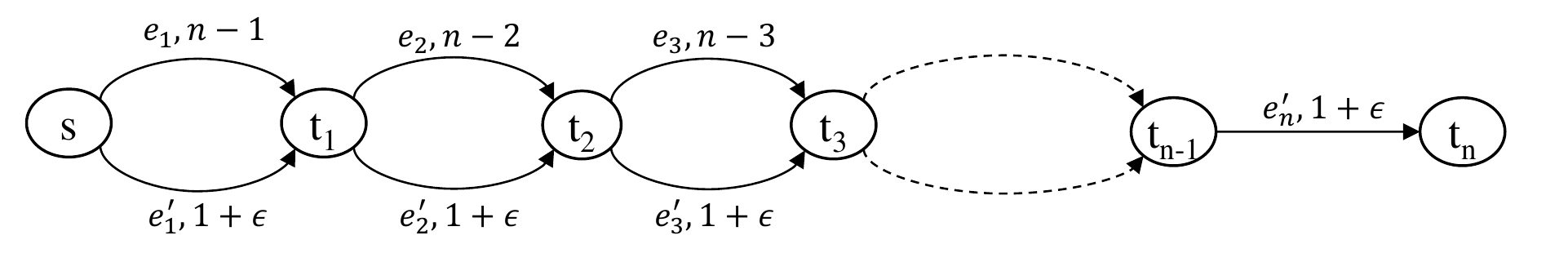}
\caption{A network on which Min-Path has inefficiency $\Omega(m)$.}
\label{fig:minpathmultt}
\end{center}
\end{figure}

We show that the analysis is tight: Consider the network depicted in Figure \ref{fig:minpathmultt}. There are $n = m$ players, where $t_i$ is the target of player $i$.
In the initial strategy profile for every $1 \leq i \leq n$, $p^{0}_{ i }=\langle {e_1}, {e_2}, \ldots , {e_{i-1}},{e^{\prime}_i} \rangle$. Note that in every segment, $E_j$, connecting $t_{ j - 1 }$ and $t_j$, the upper edge costs $n - j$ and is used by the $n - j$ players ${j+1},...,n$ and the lower edge costs $1 + \epsilon$ and is used only by player $j$.
Thus, the initial cost of player $i$ is $c_{p^0_{i}}(p^0)= \underset{{1 \leq j \leq i-1}}{ \sum } \frac{ n - j }{ n - j } + (1 + \epsilon) = ( i - 1 ) + (1 + \epsilon) = i + \epsilon$.

Clearly, in every segment, the players using the upper edge will benefit from deviating to the lower one and the player using the lower edge will benefit from deviating to the upper one. By Lemma \ref{lem:1st}, the first deviation will determine the edge that will be used in the NE reached.

Note that player $i$'s best response, is $\langle {e^{\prime}_1}, {e^{\prime}_2}, \ldots, {e^{\prime}_{i-1}} , {e_i} \rangle$ whose cost is $\underset{ 1 \leq j \leq i - 1 }{ \sum }( 1 + \epsilon ) + ( n - i ) = ( i - 1 ) \cdot ( 1 + \epsilon ) + ( n - i ) = n + ( i - 1 ) \cdot \epsilon$. Player $1$ will be chosen by the Min-Path deviator rule, and will deviate to $\langle {e_1} \rangle$. After this deviation ${e_1}$ is included in the BR of all the players. We can therefore consider the game induced by the nodes $\{ t_1 , ... , t_n \}$, treating $t_1$ as the source. The resulting game is similar to the initial one, and the same analysis shows that the second deviation is of player $i_2$ to $\langle {e_1} , {e_2} \rangle$. Repeating the process, we get that NE reached by Min-Path consists of the edges $\{ {e_1} , {e_2} , \ldots, {e_{n-1}} , {e^{\prime}_n} \}$ and therefore $SC(NE_{Min-Path}(p^0)) = \underset{ 1 \leq j \leq (n-1) }{ \sum } c_{e_j} + c_{e^{\prime}_n} = \underset{ 1 \leq j \leq (n-1) }{ \sum } ( n - j ) + 1 + \epsilon = (n - 1) \frac{ n }{ 2 } + 1 + \epsilon$.

On the other hand, another valid BR sequence is a one in which player $n$ migrates first. Her BR is deviating to all the lower edges. By Lemma \ref{lem:1st}, other players will follow and the BR sequence will converge to $\{{e^{\prime}_1} , {e^{\prime}_2}, \ldots, {e^{\prime}_n} \}$ having $SC = \underset{ 1 \leq j \leq n }{ (1 + \epsilon ) } = n(1 + \epsilon )$. The ratio between the two social costs is $\frac{ (n - 1) \frac{ n }{ 2 } + 1 + \epsilon }{ n(1 + \epsilon ) } \underset{ \epsilon \rightarrow 0 }{ \rightarrow } \frac{ (n - 1) \frac{ n }{ 2 } + 1 }{ n } \approx \frac{n}{2}$. Since $n=m$, we conclude that the inefficiency of Min-Path in SPP networks with single-source and multi-targets is $\theta(m)$.
\end{proof}

We turn to consider arbitrary instances and show that Min-Path can perform poorly.
\begin{theorem}
\label{thm:MP_SPP_ARB}
The inefficiency of Min-Path in SPP network formation games with multi-sources and multi-targets is $\theta( 2^m )$.
\end{theorem}
\begin{proof}
\label{lem:playerBoundOpt}
Let $c(R_i)$ denote the total cost of resolved segments after $i$ deviations of
players whose deviation resolved at least one segment. The segments that are resolved already in $R_0$ will be included in the NE reached, therefore, $c(R_0) \leq OPT$. We
prove that $c(R_i) - c(R_{i-1}) \leq c(R_{i-1})+OPT$ for every $i$; i.e., $c(R_i) \leq 2c(R_{i-1})+OPT$. This implies that the total network's cost is $c(R_m) \leq 2^{m+1}\cdot OPT$.

Let $i$ be the $i^{th}$ player chosen to deviate by Min-Path that has some unresolved segments. The total cost of the unresolved segments that $i$ resolves is $c(R_i) - c(R_{ i - 1 })$. Let $i^{ \prime }$ be a player guaranteed by Lemma 3. Since $i$ was chosen by Min-Path, the cost of $i$'s BR path is lower than the cost of $i^{ \prime }$'s BR path.
But the cost of $i$'s BR path is at least the cost of the unresolved segments in
$i$'s BR path. On the other hand, the cost of $i^{ \prime }$'s BR path equals the sum of the cost of her resolved segments, which is bounded by $c(R_{ i - 1 })$ and the cost she would set to her unresolved segments, bounded by $OPT$ (by Lemma 3). Putting
it all together, we get $c(R_i) - c(R_{ i - 1 } ) \leq c( BR_i( p ) ) \leq c( BR_{ i^{ \prime  } }( p ) ) \leq c(R_{ i - 1 }) + OPT$, as required. 

We show that the bound is tight: \begin{figure}
\begin{center}
\includegraphics[height=3cm]{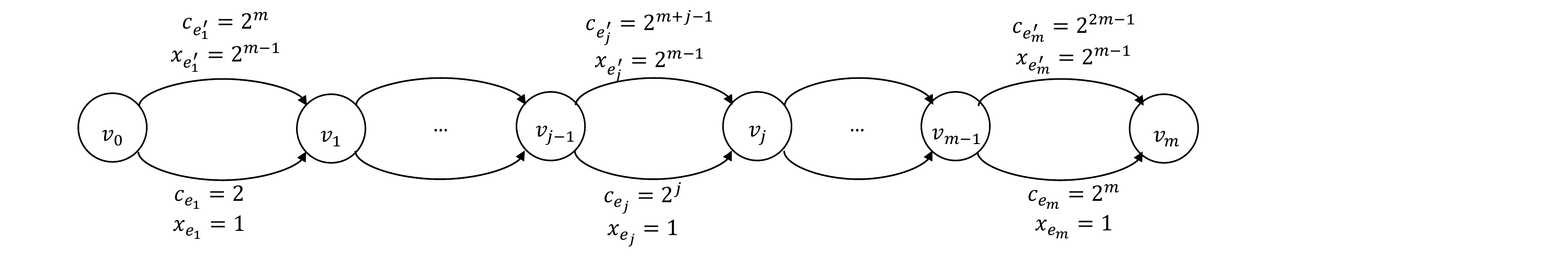}
\caption{Min-Path lower bound for multi-sources multi-targets SPP.}
{\footnotesize Every edge is labelled by its cost and the number of players using it in the initial strategy profile. E.g., the edge $e^{ \prime }_{ 1 }$ costs $2^m$ and is used by $2^{ m - 1 }$ players in the initial strategy profile.}
\label{fig:MPMULTI}
\end{center}
\end{figure}
Consider the network depicted in Figure~\ref{fig:MPMULTI}. There are $2^{ m - 1 } + m$ players as follows:
For $1 \leq i \leq m$: The source of player $i$ is $s_i = v_{i-1}$ and her target is $t_i = v_{i}$. $i$'s initial strategy is $p_{ i }^{ 0 } = e_{ i }$. Therefore, there is one player in each segment of the network who uses the lower edge. There are $2^{ m - 1 }$ other players, the "upper-players", with objective $\langle v_0, v_m \rangle$ who use the upper path $\langle e^{ \prime }_{ 1 } , ... , e^{ \prime }_{ m } \rangle$. For every segment $E_i$, the BR path of player $i$ is $\langle e^{ \prime }_{ i } \rangle$, since $\frac{ c_{ e^{ \prime }_{ i } } }{ x_{ e^{ \prime }_{ i } } + 1 } = \frac{ 2^{ m + i - 1 } }{ 2^{ m - 1 } + 1 } < 2^{ i } = \frac{ c_{ e_{ i } } }{ x_{ e_{ i } } }$ and the BR path of the upper-players is the lower path since $\frac{ c_{ e_{ i } } }{ x_{ e_{ i } } + 1 } = \frac{ 2^{ j } }{ 1 + 1 } = 2^{ j - 1 } < \frac{ 2^{ m + j - 1 } }{ 2^{ m - 1 } } = \frac{ c_{ e^{ \prime }_{ i } } }{ x_{ e^{ \prime }_{ i } } }$.

An optimal BR sequence start by a deviation of an upper-player to the lower path $\langle e_1 , ... e_m \rangle$, after that deviation all the players $1 \leq i \leq m$ do not have a beneficial deviation and all the other upper-players will follow. The NE will be $\langle e_1 , ... e_m \rangle$, whose cost $OPT = \sum_{ 1 \leq j \leq m } c_j =
\sum_{ 1 \leq j \leq m } 2^{ j } =
2^{ m + 1 } - 2$.

Consider now the BR sequence induced by the Min-Path deviator rule. Notice that for two players $1 \leq i_1 < i_2 \leq m$, the BR-paths costs satisfy $c(BR_{ i_ { 1 } }( p^0 ) ) =
c( e^{ \prime }_{ i_1 } ) =
2^{ m + i_1 - 1 } <
2^{ m + i_2 - 1 } =
c( e^{ \prime }_{ i_2 } ) =
c(BR_{ i_{ 2 } }( p^0 ) )$. The cost of the BR path of the upper-players is $\sum_{ 1 \leq j \leq m } c_{ e_j } =
\sum_{ 1 \leq j \leq m } 2^{ j } =
2^{ m + 1 } - 2$. Thus, Player $1$ whose BR path cost $2^{ m }$ will be chosen first. After her migration to $e^{ \prime }_{ 1 }$, the BR path of the upper players becomes $\langle e^{ \prime }_{ 1 } , e_{ 2 } , ... , e_{ m } \rangle$. The cost of this path is $2^m + \sum_{ 2 \leq j \leq m } 2^{ j } =
2^m + 2^{ m + 1 } - 2 > 2^{ m + 1 } =
c_{ e^{ \prime }_{ 2 } } =
c( BR_{ 2 }( p^{ 1 } ) )$, and therefore the second player to migrate will be Player $2$.

Continuing in the same manner, the first $i < m$ migrations are of players $1 , ... , i$ and the BR path of the upper-players is $\langle e^{ \prime }_{ 1 } , ... , e^{ \prime }_{ i } , e_{ i + 1 } , ... , e_{ m } \rangle$, having cost $\sum_{ 1 \leq j \leq i } 2^{ m + j - 1 } + \sum_{ i+1 \leq j \leq m } 2^j =
( 2^{ m + i } - 1 - \sum_{ 1 \leq j \leq m-1 } 2^j ) + ( 2^{ m + 1 } - 1 - \sum_{ 1 \leq j \leq i } 2^{ j } ) =
( 2^{ m + i } - 1 - 2^m + 1 ) + ( 2^{ m + 1 } - 1 - 2^{ i + 1 } + 1 ) = 2^{ m + i } + 2^{ m } - 2{ i + 1 } >
2^{ m + i }$. Player $i+1$ is the minimal of the first $m$ players that did not deviate and her BR path's cost is $c_{ e^{ \prime }_{ i + 1 } } = 2^{ m + i }$, so the Min-Path deviator rule will choose player $i + 1$.

We conclude by induction that Min-Path deviator rule will select the players in order $1 , ... , m $, and the NE that will be reached is $\langle e^{ \prime }_{ 1 } , ... , e^{ \prime }_{ m } \rangle$, having $SC(NE_{ Min-Path }( p^0 )) =
\sum_{ 1 \leq j \leq m } 2^{ m + j - 1 } =
2^{ m - 1 } \cdot \sum_{ 1 \leq j \leq m } 2^{ j } =
2^{ m - 1 } \cdot ( 2^{ m + 1 } - 1 ) =
2^{ 2m } - 2^{ m - 1 }$. The inefficiency of Min-Path in this game is $\frac{ 2^{ 2m } - 2^{ m - 1 } }{ 2^{ m + 1 } - 2 } =
\Omega( 2^m )$.

\end{proof}

\subsection{Local Rules for Extension Parallel Graphs }
\label{sec:ep}
We now show that the inefficiency of any local deviator rule is as high as the PoA, namely, $\Omega(n)$, even in the restricted class of EP networks.
Recall that the state vector of a player consists of the player's cost in her current profile and in the profile obtained by a deviation of the player, and the total cost of the path used by the player in the two profiles.

\begin{theorem}
\label{thm:feature-based}
For the class of single-source network formation games played on extension-parallel networks, the inefficiency of every local deviator rule is $\Omega(n)$.
\end{theorem}

%


\begin{proof}
Consider the network depicted in Figure \ref{fig:feature-based}(a).
There are $n$ players, all sharing the source $s$, and the targets are as depicted in the figure.
Consider the following profile:
\begin{itemize}
\item Player $1$ uses the path $\langle e_1 \rangle$.
\item Player $2$ uses the path $\langle e_1,e_2 \rangle$.
\item Players $3,4$ use the path $\langle e_3 \rangle$.
\item Players $5, \ldots, n$ use the path $\langle e_5 \rangle$. 
\end{itemize}


\begin{figure}[ht]
\begin{center}
\includegraphics[height=3.0cm]{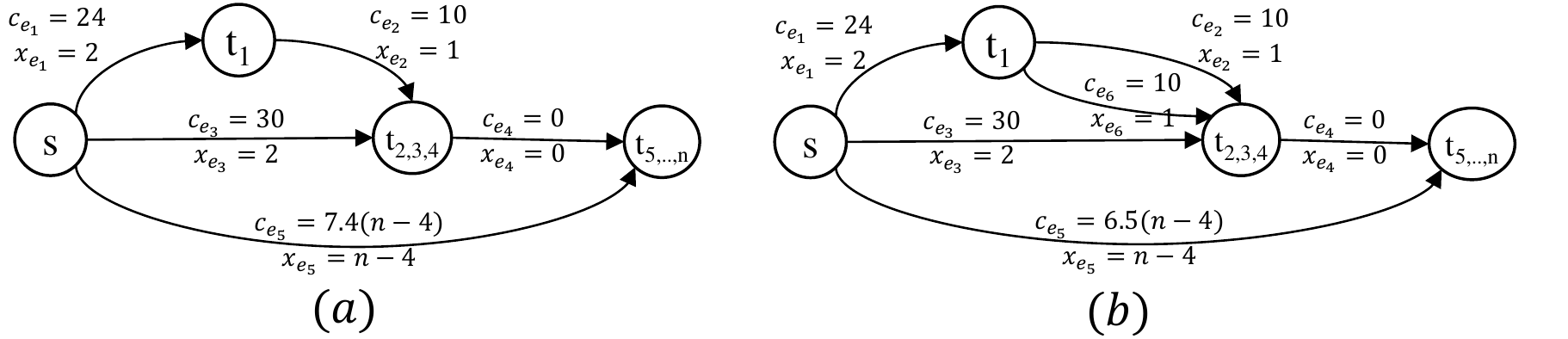}
\caption{A local deviator rule fails $(a)$ if $v_2$ is preferred, and $(b)$ if $v_3$ is preferred.}
{\footnotesize Every edge is labelled by the edge cost and the number of players using it in the initial strategy profile. E.g., the edge $e_{ 1 }$ costs $24$ and is used by two players in the initial strategy profile.}
\label{fig:feature-based}
\end{center}
\end{figure}

Consider player $2$, who uses the path $\langle e_1,e_2 \rangle$.
Her current cost is $22$ (she shares the cost of edge $e_1$ with player $1$ and pays fully for edge $e_2$), the total cost of her path is $34$, her post-deviation cost is $10$ (obtained by deviating to $e_3$, and sharing this cost with players $3,4$), and the total cost of her post-deviation path is $30$.
Thus, the state vector of player $2$ is $v_2 = (22, 34, 10, 30 )$.
Similarly, one can verify that the state vector of player $3$ (or player $4$) is $v_3 = (15, 30, 13, 34)$ (obtained by deviating to the path $\langle e_1,e_2 \rangle$).
The suboptimal players in this profile are players $2$ and $3$ (or $4$).
If the deviator rule chooses the state vector $v_2$ over $v_3$, then player $2$ will deviate to $e_3$, reaching a NE whose social cost is $54 + 7.4 (n-4)$.
On the other hand, if the deviator rule chooses the state vector $v_3$ over $v_2$, then player $3$ will deviate to $\langle e_1 , e_2 \rangle$, and from this point on all players will deviate to $\langle e_1 , e_2 \rangle$, reaching a NE whose social cost is $34$.
We conclude that a deviator rule that prefers $v_2$ over $v_3$ reaches an inefficiency of $\frac{54+7.4(n-4)}{34}=\Omega(n)$.



Consider next the network depicted in Figure \ref{fig:feature-based}(b).
There are $n$ players, all sharing the source $s$, and the targets are as depicted in the figure.
Consider the following profile:
\begin{itemize}
\item Player $1$ uses the path $\langle e_1, e_2 \rangle$.
\item Player $2$ uses the path $\langle e_1, e_6 \rangle$.
\item Players $3,4$ use the path $\langle e_3 \rangle$.
\item Players $5, \ldots, n$ use the path $\langle e_5 \rangle$.
\end{itemize}


One can verify that the state vector of player $2$ (or $1$) is $v_2 = (22, 34, 10, 30 )$ (where the best response is a deviation to $\langle e_3 \rangle$), and the state vector of player $3$ (and $4$) is $v_3 = (15, 30, 13, 34)$ (where the best response is a deviation to $\langle e_1, e_2 \rangle$). These are the only suboptimal players.
By the previous scenario, in order to avoid an inefficiency of $\Omega(n)$, the deviator rule must choose $v_3$ over $v_2$.
If player $3$ deviates, then she deviates to $\langle e_1, e_2 \rangle$. After this deviation, players $4$ and $2$ will also deviate to $\langle e_1, e_2 \rangle$, but the $n-4$ players will stay in their original path (for a cost of $6.5$ each, compared to $6.8$ upon deviation), resulting in a social cost of $34 + 6.5 (n-4)$.
In contrast, a deviator rule that chooses $v_2$ over $v_3$ will result in players $1$ and $2$ deviating to $\langle e_3 \rangle$, followed by the $n-4$ bottom players deviating to $\langle e_3, e_4 \rangle$.
This leads to a social cost of $30$, so a deviator rule that prefers $v_3$ over $v_2$ results in an inefficiency of $\frac{34+6.5(n-4)}{30}=\Omega(n)$.
We conclude that any local deviator rule has inefficiency of $\Omega(n)$.
\end{proof}


\subsection{Weighted Symmetric Network Formation Games on Parallel-Edge Graphs}
\label{sec:weightedSymm}
A {\em weighted symmetric resource selection network formation game} \cite{ARV09,HK12}, also known as a network formation game with weighted players on parallel $\langle s, t \rangle$ links, is a game in which the players are weighted, and all players have the same set of singleton strategies. Formally, each player $i$ has a {\em weight} $w_i>0$, and her contribution to the load of the (single) resource she uses as well as her payment are multiplied by $w_i$, i.e., if an edge with cost $c_e$ is shared by $k$ players with weights $w_1, w_2 ,\ldots, w_k$ then player $i$ pays $\frac{ w_i }{ \sum_{j=1}^{k}w_j } \cdot c_e$.

In Section \ref{sec:weightedHardness} we prove that it is NP-hard to find a deviator rule with inefficiency lower than $\frac{3}{2}$, and in Section \ref{sec:weightedNoLocal} we show that no local rule can guarantee a constant efficiency. In Section \ref{sec:weightedMinPath} we show that while being optimal for unweighted games, Min-Path has inefficiency $n$ in weighted symmetric games, even if the weights are arbitrarily close to each other.

\subsubsection{Hardness Result}
\label{sec:weightedHardness}
We first consider the computational complexity of finding a good BR sequence, and show that it is NP-hard to find one, or even to achieve inefficiency at most $\frac{3}{2}$.

\begin{theorem}
\label{thm:weightedHard}
In weighted network formation games on parallel edge networks, it is NP-hard to find a deviator rule with inefficiency at most $\frac{3}{2}$.
\end{theorem}

\begin{proof}
We show a reduction from the {\em Partition} problem: Given a set of numbers $ \{ a_1 , a_2 , ... , a_n \} $ such that $\sum_{ i \in [n] } a_i = 2$, the goal is to find a subset $ I \subseteq [n] $ such that $\sum_{ i \in I } a_i = \sum_{ i \in [n] \backslash I } a_i = 1$. This problem is NP-hard even if it is known that such a subset $I$ exists. Given an instance of {\em Partition}, such that for some $ I \subseteq [n] $, it holds that $\sum_{ i \in I } a_i =1$, construct the network and initial profile depicted in Figure \ref{fig:STNPH}(a). Specifically, the instance consists of $n+8$ players using four parallel links as follows:
\begin{itemize}
\item $e_1$ has a large cost $C$ and is used by one player of weight $2$, denoted {\em the $2$-player}.
\item $e_2$ has cost $3 + \epsilon$ and is shared by $n$ players with the weights $\{ a_i \}_{i=1}^n$ corresponding to the Partition instance.
\item $e_3$ has cost $9 + \epsilon$ and is shared by $6$ unit-weight players.
\item $e_4$ has cost $2$ and is not used by any player.
\end{itemize}

\begin{figure}
\begin{center}
\includegraphics[height=3.75cm]{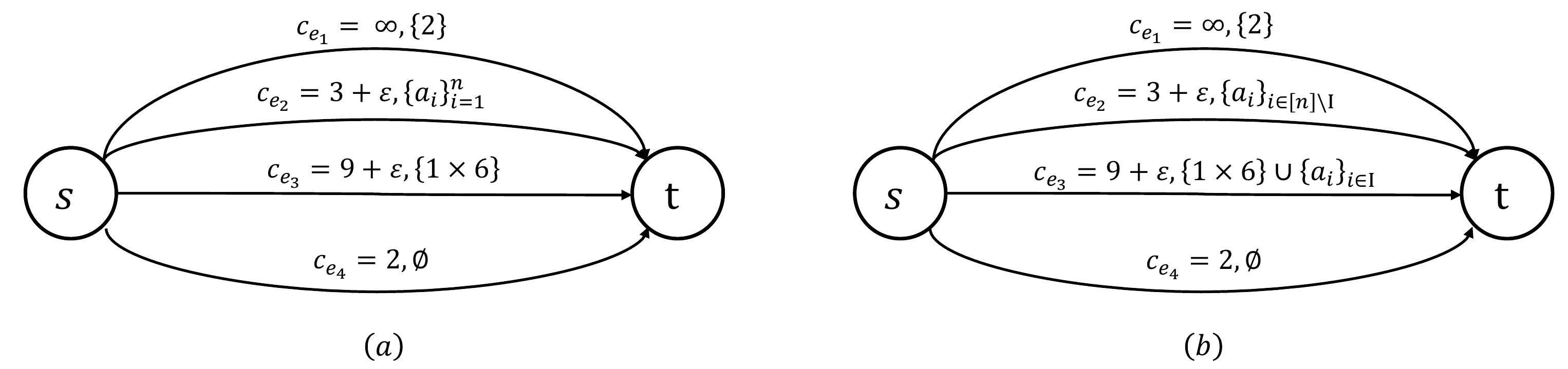}
\caption{The network constructed in the reduction from Partition.}
{\footnotesize $(a)$ is the initial configuration, and $(b)$ is the configuration before the migration of the $2$-player.
The labels on an edge give its cost, and the set of weights of the players assigned to it. E.g., the cost of edge $e_3$ in profile $(a)$ is $9 + \epsilon$ and it is used by $6$ players of weight $1$.}
\label{fig:STNPH}
\end{center}
\end{figure}
The following claim completes the proof.
\begin{claim}
\label{clm:partition}
Let $I$ be a subset of $[n]$ such that $\sum_{ i \in I } a_i=1$. A deviator rule that knows $I$ can determine a BR-sequence that ends up with a NE whose social cost is $2$. If $I$ is not detected then the final NE will have social cost more than $3$.
\end{claim}
\begin{proof}
Assume that a set $I \subset [ n ]$ s.t. $\sum_{ i \in I } a_i = 1$ is known. Observe the BR-sequence that starts by migrations of the players in $I$, and then the player with weight $2$. Note that migrating to $e_4$ incurs a cost of $2$ which is more than their current cost $\frac{ c_2 }{ l_2 } \cdot a_i = \frac{ 3 + \epsilon }{ 2 } \cdot a_i < 2$, for $a_i < 1$.
The BR of the first player is $e_3$ since $\frac{ 9 + \epsilon }{ 6 } < \frac{ 3 + \epsilon }{ 2 }$. Also, after the first player it would be even more beneficial to the others to follow. After the players in $I$ deviate, the optimal BR sequence let the player with weight $2$ migrate.
The load on $e_2$ at this time point is $1$, and the load on $e_3$ is $7$. So the costs the $2$-player would incur by deviating to $e_2,e_3$ and $e_4$ are $\frac{ 3 + \epsilon }{ 3 } \cdot 2, \frac{ 9 + \epsilon }{ 9 } \cdot 2$ and $2$, respectively.
Each of the first two alternatives is strictly larger than $2$, and therefore, the $2$-player will migrate to $e_4$. It is easy to verify that all the unit-weight players will follow and then all the other players. The NE that will be reached is the edge $e_4$ which is clearly optimal.

For the other side of the reduction proof, assume that there exists a BR sequence converging to a NE with $SC=2$. Such a sequence must converge to $e_4$. Our proof is based on analyzing the conditions required for having $e_4$ the BR of some player for the first time. Specifically, we show that the $2$-player is the only one whose BR move may utilize $e_4$, and that this can happen if and only if the loads on $e_2$ and $e_3$ are exactly $1$ and $7$ respectively. First note that as long as $e_4$ is not utilized, each of the players on $e_2$ and $e_3$ either pays less than $2$ or has a BR with cost that is less than $2$.
Therefore, $e_4$ can be the BR only of the $2$-player. Denote by $l_2$ and $l_3$ the loads (total weight) on $e_2$ and $e_3$ when the $2$-player performs BR and migrate to $e_4$.
Given that $e_4$ is the BR of the $2$-player it must hold that $2 \cdot\frac{ c_2 }{ l_2 + 2 } > 2$ and $2 \cdot\frac{ c_3 }{ l_3 + 2 } >2$.
Substituting the costs of $e_2$ and $e_3$, we get $l_2<1+\epsilon$ and $l_3 < 7+\epsilon$. Since $\epsilon$ can be arbitrarily small, as shown in Figure \ref{fig:STNPH}(b), these conditions can only be fulfilled if players of total weight exactly $1$ have migrated from $e_2$ to $e_3$. This set of players corresponds to a subset $I$ in the Partition instance whose total sum is $1$. Thus, converging to a NE with $SC=2$, must involve a detection of the subset $I$. Moreover, since the second-best NE has $SC=3 + \epsilon$, we conclude that any deviator rule with inefficiency $\frac{3}{2}$, must be able to exactly solve the partition problem.
\end{proof} \end{proof}

\subsubsection{ Local Deviator Rules }
\label{sec:weightedNoLocal}
Our bad news continue with a negative result referring to any local deviator rule. Recall that in weighted network formation games, the state vector of a player includes her weight, her current cost, her strategy's cost, her cost after performing BR, and the cost of her BR strategy.

\begin{theorem}
\label{thm:weightedLB}
Any local deviator rule has inefficiency $\Omega ( \sqrt{ n } )$ in weighted network formation games on parallel-edges networks.
\end{theorem}
\begin{proof}
Given $r$, we show that no local deviator rule has inefficiency lower than $r/2$. Consider the game $G_a$ depicted in Figure \ref{fig:weighted2}(a).
\begin{figure}
\begin{center}
\includegraphics[height=3.5cm]{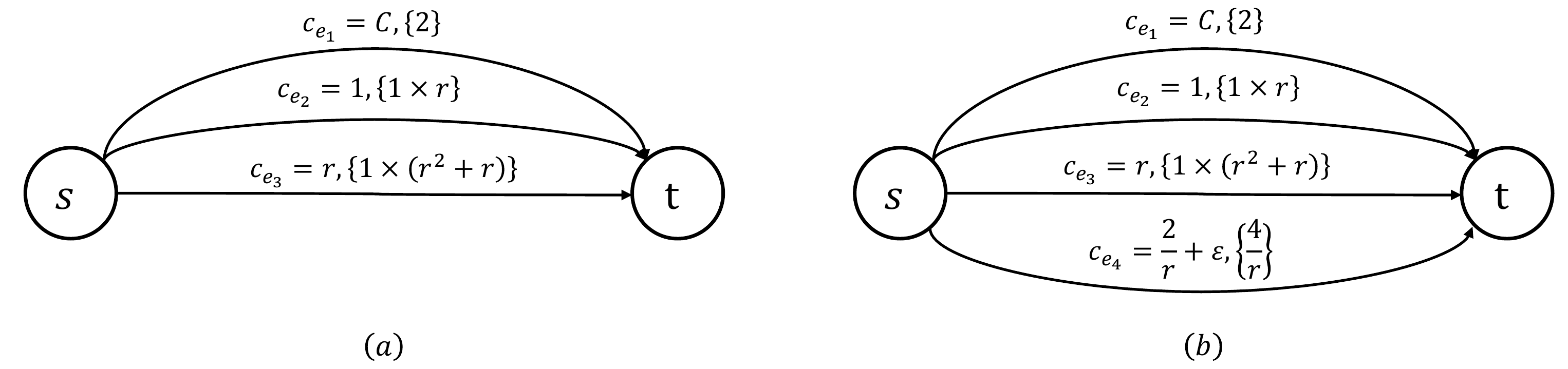}
\caption{A feature-based deviator rule fails $(a)$ if $v_2$ is prioritized, and $(b)$ if $v_1$ is prioritized.}
{\footnotesize The labels on an edge give its name, its cost, and the set of weights of the players assigned to it in $p^0$.
The labels on an edge give its cost, and the set of weights of the players assigned to it. E.g., the cost of edge $e_3$ in profile $(a)$ is $r$ and it is used by $( r^2 + r )$ players of weight $1$.}
\label{fig:weighted2}
\end{center}
\end{figure}

The network consists of three parallel edges. The upper edge, $e_1$ has a very high cost, $C$, and is used only by Player $1$ whose weight is $w_1 = 2$. The middle edge, $e_2$ costs $1$ and is used by $r$ unit-weight players $2 , ... , r + 1$.
The lower edge, $e_3$ costs $r$ and is used by $r^2 + r$ unit-weight players $r + 2 , ... , r^2 + 2r + 1$, each has weight $w_i = \frac{ r }{ r + 1 }, i \in \{ 2 , ... , r + 1 \}$. The suboptimal players in $p^0$ are the players on $e_1$ and $e_2$.
Since $\frac{ 1 }{ r+2 } < \frac{ r }{ r^2 + r + 2 }$ the BR of Player 1 is $\langle e_2 \rangle$. Also, since $\frac{ r }{ r^2 + r + 1 } < \frac{ 1 }{ r }$, a deviation to $\langle e_3 \rangle$ is the BR of Player 2 and all its equivalents. Observe the state vectors corresponding to Player 1 and Player 2: Player 1 currently pays $C$ for a path that costs $C$, and after a deviation to its BR path $e_2$ will pay $\frac{ 2 }{ r + 2 }$ for a path that costs $1$. Her weight is $2$ and therefore $v_1=(C,C, \frac{ 2 }{ r+2 }, 1 , 2)$. One can verify that the corresponding state vector of Player 2 is $v_2=(\frac{1}{r}, 1 , \frac{ r }{ r^2+r+1 }, r , 1 )$ (obtained by a deviation to $e_3$).


If a deviator rule prefers $v_1$ over $v_2$ then the resulting NE will be $\langle e_2 \rangle$, since no player on $e_2$ will be suboptimal after the first deviation. On the other hand, if $v_2$ is chosen then after her deviation the only BR path is $e_3$ and the resulting NE will be $\langle e_3 \rangle$. We conclude that in order to achieve inefficiency less than $r$, the deviator rule must choose $v_1$ over $v_2$.

Consider now the game $G_b$ depicted in Figure \ref{fig:weighted2}(b). It extends the game $G_a$ with an additional edge $e_4$ whose cost is $\frac 2 r +\epsilon$ and a single player $r^2 + 2r + 2$ having weight $w_{ r^2 + 2r + 2 } = \frac{4}{r}$. We show that in this game, a deviator rule must choose $v_2$ in order to reach an efficient NE.
Note that for sufficiently large $r$ and respectively sufficiently small $\epsilon$,
\begin{displaymath}
\frac{ \frac{ 2 }{ r } + \epsilon }{ \frac{ 4 }{ r } + 1 } =
\frac{ 2 + \epsilon \cdot r }{ r + 4 } >
\frac{ r }{ r^2 + r + 1 }
~~~\mbox{~~and~~}~~~~~~~
\frac{ \frac{ 2 }{ r } + \epsilon }{ \frac{ 4 }{ r } + 2 } =
\frac{ 1 + \epsilon \cdot \frac{ r }{ 2 } }{ r + 2 } > \frac{ 1 }{ r + 2 }.
\end{displaymath}
Therefore, the new edge $e_4$ is not the BR choice of any player and the state vectors of the players on $e_1$ and $e_2$ remain $v_1$ and $v_2$ as in the game $G_a$.
We show that an efficient NE is achieved if and only if a player on $e_2$ plays first.

If the $\frac{ 4 }{ r }$-player is chosen to deviate first, she will select $\langle e_3 \rangle$ as her best response, since $\frac{ r }{ r^2 + r + \frac{ 4 }{ r } } = \frac{ 1 }{ r + 1 + \frac{ 4 }{ r^2 } } < \frac{ 1 }{ r + \frac{ 4 }{ r } }$. After this deviation $\langle e_4 \rangle$ is the least profitable choice, and therefore a NE consisting of $\langle e_4 \rangle$ will not be reached.
The players on $e_3$ are not suboptimal in $p^0$, and therefore, will never be selected to play first.
If $v_1$ is chosen, then similar to $G_a$, we will end up with a NE consisting of $\langle e_3 \rangle$.

An optimal BR sequence chooses $v_2$ and let a unit player from $e_2$ migrate to $e_3$. Then, Player $1$ is selected. Since
$
\frac{ \frac{ 2 }{ r } + \epsilon }{ \frac{ 4 }{ r } + 2 } =
\frac{ 1 + \epsilon \cdot \frac{ r }{ 2 } }{ r + 2 } <
\min \{ \frac{ 1 }{ r - 1 + 2 } , \frac{ r }{ r^2 + r + 1 + 2 } \}
$
she will choose the lower cheapest path $\langle e_4 \rangle$. The unit-players will then follow, leading to $SC = \frac{ 2 }{ r }$.

We conclude that in order to avoid inefficiency $\frac{ r }{ 2 }$, in each of the games, an optimal deviator rule must choose a different deviator among the players corresponding to state vectors $v_1$ and $v_2$. Therefore, any local deviator rule has inefficiency at least $\frac{ r }{ 2 } = \Omega ( \sqrt{ n } )$.
\end{proof}


\subsubsection{ The Performance of Min-Path in Weighted Network Formation Games }
\label{sec:weightedMinPath}
Weighted network formation games on parallel-edges networks can be extended to weighted symmetric games with strategies consisting of two resources (a $2$-segment SPP). By \\ \cite{AD+08}, these are potential games and any BR sequence converges to a NE.
In Theorem \ref{thm:minopt} we showed that the deviator rule Min-Path ensures the optimal reachable NE in the
unit-weight symmetric games. We show that with weighted players, the inefficiency of Min-Path can be as bad as the PoA even if the weights are arbitrarily close to each other, and the strategies are sets of two resources.

\begin{theorem}
\label{thm:minweightedbad}
The Min-Path deviator rule in weighted network formation games on a $2$-segment SPP have inefficiency $\Omega ( n )$. This holds even if the ratio $\max_i w_i/ \min_i w_i$ is arbitrarily close to $1$.
\end{theorem}
\begin{proof}
Let $k > 2$ be an integer, and let $l = 1 + \frac{ 2 }{ k }$ and $r = l \cdot k = k + 2$. Consider the SPP network consisting of two segments, $E_1$ and $E_2$, depicted in Figure~\ref{fig:noRmin1}.
The game consists of $r-1 + \frac{ r }{ l } = 2k + 1$ players in the following initial configuration:
\begin{itemize}
\item $\frac{ r }{ l } = k $ players with weights $l$, denoted {\it $l$-players}. Each of these players is using an edge with a cost of $2r$ that only she is using in $E_1$ and the upper edge in $E_2$.
\item $r - 1$ unit-weights players. Each of these players is using the lower edge in $E_1$ and an edge of cost $2r$ that only she is using in $E_2$.
\end{itemize}

\begin{figure}
  \begin{center}
\includegraphics[height=4.357cm]{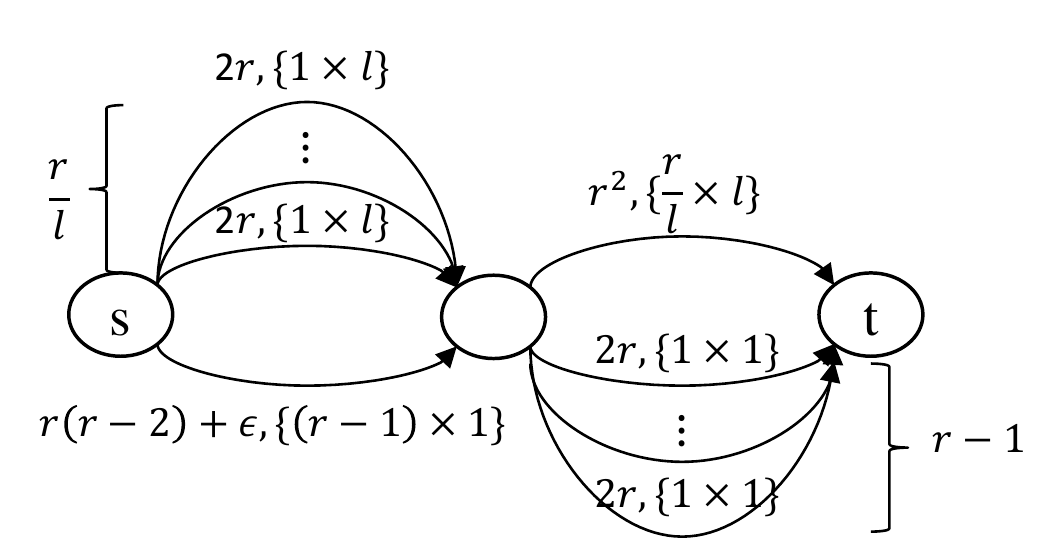}
  \caption{\small An instance for which Min-Path ensures no finite inefficiency.}
{\footnotesize Every edge is labelled by its cost and the weights of players using it.}
\label{fig:noRmin1}
  \end{center}
\end{figure}
We calculate the players' BR in this configuration.
The BR edge of the $l$-players in $E_1$ is the lower one since for sufficiently small $\epsilon$:
\begin{displaymath}
\frac{ 2r }{ 2l } \cdot l =
r >
r \cdot \frac{ k + 2 }{ k + 2 + \frac{ 2 }{ k } } + \epsilon \frac{ l }{ ( r - 1 ) + l } =
r \cdot \frac{ l \cdot k }{ ( r - 1 ) + l } + \epsilon \frac{ l }{ ( r - 1 ) + l } =
\frac{ r( r - 2 ) + \epsilon }{ ( r - 1 ) + l } \cdot l.
\end{displaymath}
The BR edge of the $l$-players in $E_2$ is one of the lower edges since
$\frac{ r^2 }{ r } \cdot l = r \cdot l >
\frac{ 2r }{l +1} \cdot l$.
Thus, the BR path of every $l$-player is migrating to lower edges, and therefore the cost of their BR path is $r(r-2) + \epsilon + 2r = r^2 + \epsilon$.
The unit-weight players' BR edge in $E_1$ is an upper edge since
\begin{displaymath}
\frac{ r( r - 2 ) + \epsilon }{ r - 1 } >
\frac{ r( r - 2 ) }{ r - 1 } =
r \cdot \frac{ k }{ k + 1 } =
r \cdot \frac{ 2 }{ 1 + 1 + \frac{ 2 }{ k } } =
\frac{ 2r }{ l + 1 }.
\end{displaymath}
The BR edge in $E_2$ for a unit-weight player is the upper one since
$\frac{ r^2 }{ r + 1 } < \frac{ 2r }{ 1 + 1 }$.
Thus, the BR path of the unit-weight players is the upper edges, and therefore the cost of their BR path is $2r + r^2$.
For $\epsilon < r$ it holds that $2r + r^2 > r^2 + \epsilon$; therefore, the Min-Path deviator rule will let an $l$-player whose BR is minimal migrate first. It is easy to see that the unit-weight players will join this path and the NE reached by Min-Path consists of the lower edge in the left segment and one of the lower edges in the right segment. That implies, $SC( NE_{Min-Path}( p^0 ) ) = r^2 + \epsilon$.

We show that a better BR sequence exists.
Consider a BR sequence that starts by a migration of a unit-weight player, followed by a migration of all the $l$-players whose edge in the first segment was not the one migrated to. After the migration of the unit-weight player the loads on the edges she migrated to are $l + 1$ in $E_1$, and $r + 1$ in $E_2$. It is easy to verify that for sufficiently large $k$,
\begin{displaymath}
\frac{ 2r }{ 2l + 1 } \cdot l =
r \cdot \frac{ 2l }{ 2l + 1 } =
r \cdot \frac{ 2k + 4 }{ 3k + 4 } <
r \cdot \frac{ k }{ k + 1 + \frac{ 2 }{ k } } =
r \cdot \frac{ r - 2 }{ r - 2 + l } <
\frac{ r ( r - 2 ) + \epsilon }{ r - 2 + l }
\end{displaymath}
so all of the $l$-players will follow the first migration in $E_1$. Also, the BR edge for an $l$-player in $E_2$ is a lower edge, since
\begin{displaymath}
\frac{ 2r }{ 1 + l } =
r \cdot \frac{ 1 }{ 1 + \frac{ 1 }{ k } } <
r \cdot \frac{ 1 }{ 1 + \frac{ 1 }{ k + 2 } } =
r \cdot \frac{ r }{ r + 1 } =
\frac{ r^2 }{ r + 1 }.
\end{displaymath}

It is easy to verify that after the migration of an $l$-player, her chosen path, of cost $2r + 2r$, will remain the BR path of all the $l$-players and afterwards it becomes the BR path of the unit-weight players as well. Therefore, there exists a BR sequence that converges to a NE with social cost $4r$. We conclude,
\begin{displaymath}
\frac{ SC( NE_{ Min-Path }( p^0 ) ) }{ SC( p^{ \star }( p^0 ) ) } \geq \frac{ r^2 + \epsilon }{ 4r } > \frac{ r }{ 4 } = \frac{ k + 2 }{ 4 } > \frac{ k }{ 4 }
\end{displaymath}
Note that the number of players is $n = \frac{ r }{ l } + ( r - 1 ) < 2r = 2k + 4$.
Since $k$ can be arbitrarily large, we conclude $\alpha_{Min-Path}^G = \Omega(n)$. Thus, Min-Path ensures no finite inefficiency. Note that $l$ can be arbitrarily close to $1$. This implies that our bound is valid even if the ratio between the maximal and the minimal weights is arbitrarily close to $1$.
\end{proof}

\section{ Job Scheduling Games }
\label{sec:sched}

Job scheduling games are resource selection congestion games corresponding to scenarios in which each player controls a job and needs to select a machine to process the job.
We consider two models:
\begin{enumerate}
\item
Every job is associated with a length $w_i$ corresponding to its processing time. The load on a machine $M_j$ in a schedule $p$ is the total processing time of the jobs assigned to it. That is, $L_j( p ) = \sum_{ i: p_i = M_j } w_i$. The latency function is linear, thus, the cost of job $i$ in a schedule $p$ is $c_i(p)=L_{p_i}(p)$ \cite{Voc07}.
\item
Congestion has conflicting effects (\cite{FT12,CG11}). Jobs have unit processing time and machines have an activation cost $B$. The cost function of job $i$ in a schedule $p$ consists of two components: the load on the job's machine and the job's share in the machine's activation cost. Formally, the cost of job $i$ in a schedule $p$ where $p_i=M_j$ is $c_i( p) = L_j(p) + B/L_j(p)$.
\end{enumerate}


\subsection{ Weighted Jobs on Parallel Identical Machines }
In what follows we show that no local deviator rule has inefficiency better than the PoA (which is known to be $\frac{ 2m }{ m + 1 }$ \cite{FH79}). Recall that the state vector we use consist of the job's length, its machine and the machines' loads.

\begin{theorem}
\label{thm:lb_sched}
The inefficiency of every local deviator rule is $\frac{ 2m }{ m + 1 }$.
\end{theorem}
\begin{proof}
Assume, by a way of contradiction, that there is a local deviator rule whose inefficiency is better than the PoA.
Given $m$, let $\epsilon$ be a small constant such that $( m + \frac{1}{2} \epsilon)/( \frac{3}{2} \epsilon)$ is an integer. The proof consists of two games, one can verify that both have the same initial loads profile.
\begin{figure}
\begin{center}
\includegraphics[height=4.25cm]{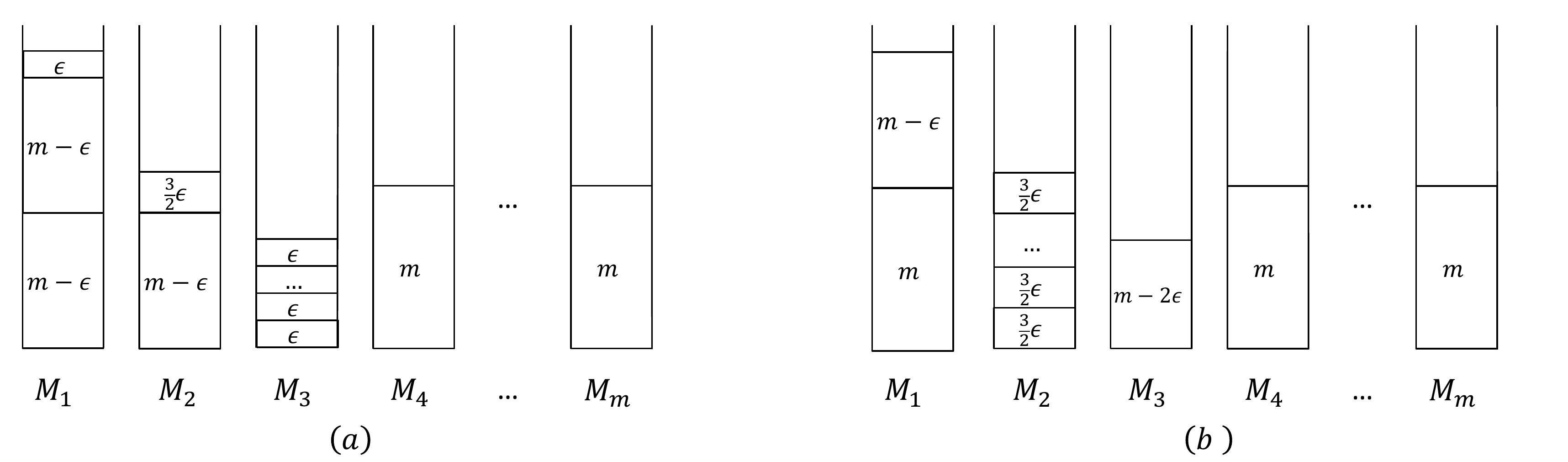}
\caption{The initial profile for which $(a)$ prioritizing $v'$, or $(b)$ prioritizing $v''$ leads to inefficiency $\frac{ 2m }{ m + 1 } = PoA$.}
\label{fig:IdenticalMachinesBoth}
\end{center}
\end{figure}

Consider first an instance with the following initial schedule, depicted in Figure \ref{fig:IdenticalMachinesBoth}$(a)$:
\begin{itemize}
\item $L_1(p^0)=2m - \epsilon$. $M_1$ processes two jobs of length $m - \epsilon$ and one job of length $\epsilon$.
\item $L_2(p^0)=m + \frac{1}{2} \epsilon$. $M_2$ processes a job of length $m - \epsilon$ and a job of length $ \frac{ 3 }{ 2 }\epsilon$.
\item $L_3(p^0)=m - 2 \epsilon$. $M_3$ processes $\frac{ m - 2 \epsilon }{ \epsilon }$ jobs of length $\epsilon$.
\item $L_j(p^0)=m$, for every machine $M_j, 4 \le j \le m$, each processing a single job of length $m$.
\end{itemize}
The only suboptimal jobs are the ones on $M_1$ and the job of length $\frac{ 3 }{ 2 } \epsilon$ on $M_2$. The state vector of each of the $(m-\epsilon)$-jobs on $M_1$ is $v'= (m - \epsilon, 1 )$.\footnote{For better readability we omit the machines' loads from the state vectors, ind include only the jobs' length and the index of its machines.} The state vector of the $\epsilon$-job on $M_1$ is $\hat{v}= (\epsilon, 1 )$. The state vector of the suboptimal job on $M_2$ is $v'' =( \frac 3 2 \epsilon, 2 )$. If a deviator rule prefers $v''$ over $v'$ then the short job on $M_2$ will migrate to $M_3$. Then, the only beneficial migration is of a job from $M_1$ to $M_2$, leading to a NE with makespan $2m - 2\epsilon$ (on either $M_1$ or $M_2$). If the deviator rule prefers $\hat{v}$ then the $\epsilon$-job on $M_1$ will deviate to reach a NE with makespan $2m - 2\epsilon$.

On the other hand, an optimal deviator rule selects $v'$, and let one of the jobs of length $m - \epsilon $ migrate to $M_3$. Then, the $\epsilon$-jobs from $M_3$ will spread among the machines, and the resulting NE would have makespan $ m + \frac{ m - \frac{5}{2} \epsilon }{ m } = m + 1 - \frac{ 5 \epsilon }{ 2 m }$.

We conclude that a deviator rule that prefers $v''$ or $\hat{v}$ has inefficiency $\frac{ 2m - 2 \epsilon }{ m + 1 - \frac{ 5 \epsilon }{ 2 m } } \underset{ \epsilon \rightarrow 0 }{ \rightarrow } \frac{ 2m }{ m + 1 }$.

Consider now the following initial schedule, depicted in Figure \ref{fig:IdenticalMachinesBoth}$(b)$.
\begin{itemize}
\item $L_1(p^0)=2m - \epsilon$. $M_1$ processes two jobs, one job of length $m - \epsilon$ and one job of length $m$.
\item $L_2(p^0)=m + \frac{1}{2} \epsilon$. $M_2$ processes $\frac{ m + \frac{1}{2} \epsilon }{ \frac{3}{2} \epsilon }$ jobs of length $\frac{ 3 }{ 2 } \epsilon$.
\item $L_3(p^0)=m - 2 \epsilon$. $M_3$ processes one job of length $m - 2 \epsilon$.
\item $L_j(p^0)=m$, for every machine $M_j, 4 \le j \le m$, each processing a single job of length $m$.
\end{itemize}

Again, the only suboptimal jobs are the ones on $M_1$ and $M_2$. If a local deviator rule chooses $v''$ then a job of length $\frac{3}{2} \epsilon$ on $ M_2 $ migrates first. It will migrate to $M_3$ and then the job of length $ m - \epsilon $ on $M_1$ will migrate to $M_2$, whose load, $m-\epsilon$ is the lightest. Then, all the jobs of length $\frac{ 3 }{ 2 } \epsilon$ on $M_2$ will spread among the machines and a NE with makaspan $m + 1 - \frac{ 5 \epsilon }{ 2 m }$ will be reached.

On the other hand, a deviator rule that does not choose $v''$, chooses a job on $M_1$. If it chooses the job of length $m - \epsilon$, whose state vector is $v'$, it would migrate to $M_3$. A NE is reached after this single migration and has makespan $2m - 2 \epsilon$. Alternatively, if the job of length $m$ is chosen first, it migrates to $M_3$ and afterwards the only beneficial migration is of the job of length $m - 2 \epsilon $ to $M_1$. Again, we will end up with makespan $2m - 2 \epsilon$.

We conclude that for this instance, the deviator rule must choose $v''$ in order to achieve inefficiency better than $\frac{ 2m }{ m + 1 }$. Since this conflicts with the optimal choice for the first instance, no local rule ensures inefficiency better than $\frac{ 2m }{ m + 1 }$.
\end{proof}

We note that this result captures several common natural deviator rules in job scheduling games, such as selecting a longest job, a job having max-cost, or a job whose BR causes the max-improvement.

\subsection{ Job Scheduling Games with Conflicting Congestions Effects }
In this section we consider scheduling of unit-size jobs on identical machines.
The machines are associated with an {\em activation cost}, $B$.
For a given profile $p$, the load on machine $j$, denoted $L_j(p)$, is
the number of jobs assigned to it.
The cost function of job $i$ in a given schedule is the sum of two components:
the load on $i$'s machine and $i$'s share in the machine's activation
cost.
The activation cost $B$ is shared equally between all the jobs assigned to a particular machine.
That is, given a profile $p$ in which $p_i = M_j$, the cost of
job $i$ is $c_i(p) = L_j(p) + \frac{B}{L_j(p)}$.

We denote the cost of a job assigned to a machine with load $x$ by $c(x)$, where $c(x) = x +\frac{B}{x}$.
As \cite{FT15} showed, the cost function exhibits the following structure.
\begin{observation}
\label{obs:optload} The function $c(x)=x+B/x$ for $x > 0$
attains its minimum at $x=\sqrt{B}$, is decreasing for $x \in
(0,\sqrt{B})$, and increasing for $x > \sqrt{B}$.
\end{observation}

Denote by $l^{*}$ the load minimizing the players' cost. Thus, $l^{ * } = \argmin_{ l \in \{ \lfloor \sqrt{ B } \rfloor , \lceil \sqrt{ B } \rceil \} } c( l ) $, with tie breaking in favor of the lower option. For a given profile $p$, let $M(p)$ denote the active machines in $p$, that is, $M(p) = \{ j : L_j( p ) > 0 \}$, and let $m(p)=|M(p)|$.
Among the machines in $M(p)$, a machine that has load at least (respectively, smaller than) $l^*$
is said to be a \emph{high} (\emph{low}) machine. By Observation \ref{obs:optload}, the BR of any migrating job, is either the most loaded low machine, or the least loaded high machine.
Since the jobs are identical, all the jobs assigned on a specific machine incurs the same cost. Thus, when analyzing BR-deviations, we assume that a deviator rule is a function  $S : P \rightarrow M(p)$ that given a profile $p$, chooses {\em a machine} rather than a job.
One arbitrary player assigned to the chosen machine then performs a BR move.
We also refer to machines, rather than jobs, as being suboptimal. A machine is suboptimal if the jobs it processes are suboptimal.
Finally, we assume that the activation cost, $B$, is a known parameter that can be used by the deviator rule.

We first show that in any NE profile, the machines are balanced. Formally:

\begin{claim}
In a job scheduling game with activation cost $B$ and unit-length jobs, in any $p \in NE$, every active machine is assigned $ \ceil{\frac{ n }{ m( p )}} $ or $ \floor{\frac{ n }{ m( p )}} $ jobs.
\end{claim}
\begin{proof}
\label{clm:balanced}
Assume by contradiction that in some NE profile $p$ the machines are not balanced, that is, there are two machines having loads $l$ and $l'$ such that $l+2 \le l'$. It must be that the machine having load $l$ is low while the machine having load $l'$ is high, as otherwise, if both are low then a migration to the higher one is beneficial, and if both are high, then a migration to the lower one is beneficial, contradicting the stability of $p$. Moreover, since $p$ is a NE, we have $c( l^{ \prime } ) \leq c( l + 1 )$ and $c( l ) \leq c( l^{ \prime } + 1 )$.
Since the cost function increases on high machines, we have $c( l^{ \prime } ) < c( l^{ \prime } + 1 )$, that is, $l^{ \prime } + \frac{ B }{ l^{ \prime } } < l^{ \prime } + 1 + \frac{ B }{ l^{ \prime } + 1 }$, implying that $B < l^{ \prime }( l^{ \prime } + 1 )$.

By the definition of the cost function, $c( l^{ \prime } ) = c(\frac B {l'})$. Recall that the cost function is decreasing for $x \le l+1$ and note that $\frac B {l'} \le l+1$. Since $c( \frac B {l'} )=c(l') \leq c( l + 1 )$, we conclude $l + 1 \le \frac{ B }{ l^{ \prime } }$. Similarly, the cost function is increasing for $x >l'$ and $\frac B {l+1}\ge l'$. Since $c( l' ) \leq c(l+1) =c( \frac B {l+1} )$ we conclude $l' \le \frac{ B }{ l+1}$. Combining both inequalities we get $\frac{ B }{ l^{ \prime } + 1 } - \frac{ B }{ l^{ \prime } } \geq 1$, which implies $B \geq l^{ \prime }( l^{ \prime } + 1 )$.
A contradiction.
\end{proof}

The following results characterize BR sequences and NE profiles.
\begin{claim}
\label{cl:max_m}
In an instance of scheduling games with conflicting congestion effects, an optimal NE is one in which the number of active machines is maximal.
\end{claim}
\begin{proof}
Let $B$ be the activation cost.
We first note that no profile with two low machines is stable - since, by Observation \ref{obs:optload}, jobs from a low machine would benefit from migrating to a more (or equally) loaded low machine.
By Claim \ref{clm:balanced}, the loads in a NE are either $ \ceil{\frac{ n }{ m( p )}} $ or $ \floor{\frac{ n }{ m( p )}} $. Therefore, increasing $m(p)$ lowers the loads and, by Observation \ref{obs:optload}, also the costs for $\frac{ n }{ m(p) } \geq l^{*}$.
If there exist $p \in NE$ such that $l^{*} - 1 < \frac{ n }{ m(p) } < l^{*}$ then some machines will have load $l^{*} - 1$ and the others will have loads $l^{ * }$. Since $p$ is stable $c( l^{ * } - 1 ) \leq c( l^{ * } + 1 )$ and therefore any NE with less machines will have some machines with load at least $l^{ * } + 1$ and therefore higher social cost. If $\frac{ n }{ m(p) } \leq l^{ *} - 1$ then it must be that $n < l^{ * }$ and the only NE is on a single machine. If there is more than a single active machine then there are two low machines, and, as we already pointed out, $p$ cannot be stable.
\end{proof}

In the following analysis we assume that the ordering of machines by load relations is fixed, that is, if $L_{ j_1 }( p^0 ) \leq L_{ j_2 }( p^0 )$ then throughout dynamics the same order remains. This assumption is w.l.o.g., since we can always relabel the machines. To avoid this relabelling we assume that whenever a job migrates to a machine having load $x$ it joins the machine having highest index among the machines having load $x$.
We begin with a simple observation.
\begin{observation}
\label{observation:HighStays}
If a machine $M_j$ is high in the initial profile, 
then for every $ p \in NE( p^0 ) $, $M_j \in M(p)$.
Also, if a job migrates to a machine $M_j$ during a BR sequence, then $M_j$ will be active in the final profile.
\end{observation}
\begin{proof}
An active machine is not in $M(p)$ if it is emptied out. However, the load of a high machine will never decrease below $l^*$ since this load incurs the minimum possible cost, and none of the jobs assigned to $M_j$ when its load is $l^*$ is suboptimal.
For the second part, let $M_j$ be a BR machine. If $M_j$ is high, then it will not be emptied out as we just showed. If $M_j$ is low, then since the cost function is decreasing for load at most $l^*$, then $M_j$ will remain the best response until it becomes high, and will not be emptied out once it reaches load $l^*$.
\end{proof}

The next lemma states a sufficient and necessary condition for a low machine to remain active in the NE profile.
We assume by relabelling that $L_1( p^0 ) \leq L_2( p^0 ) \leq ... \leq L_m( p^0 )$.
Denote $l_{i}=L_i(p^0)$.
\begin{claim}
\label{clm:lastingFormula}
For every initial profile $p^0$, and every $M_j \in M(p^0)$, there exists $p \in NE(p^0)$ in which $M_j \in M(p)$ iff $j=m$ or $\frac{ n - l_{j} }{ m - j } > \frac{ B }{ l_{ j } + 1 }$.
\end{claim}
\begin{proof}
Note that $M_m$ is the most loaded machine in $p^0$. If it is high in $p^0$ then by Observation \ref{observation:HighStays}, it will remain active. If it is low, then $p^0$ consists only of low machines, and by our relabeling assumption, at least $M_m$ will remain active.

We turn to analyze a machine $M_j$ for $j<m$.
We first show that if $M_j$ fulfills the condition, then there exists a BR sequence in which it is not emptied out. Let $\bar{l} \equiv \frac{n-l_{j}}{m - j} > \frac{B}{l_{j}+1}$. If $M_j$ is high then the claim follows from Observation \ref{observation:HighStays}. Assume that $M_j$ is low in $p^0$. Note that all the machines $M_{ lower } = \{ M_{ 1 } , ... , M_{ j - 1 } \}$ are suboptimal because they are low and also lower than $M_j$, thus, their jobs can benefit from migrating to $M_j$.
Consider a BR sequence that starts by migrations out of $M_{lower}$.
If a job chooses to migrate to $M_j$, then, by Observation \ref{observation:HighStays} we are done. Otherwise, after the machines in $M_{lower}$ are empty, the BR sequence proceeds by selecting higher machines. Note that along this sequence, there is always some higher machine $M_k$ for $k>j$ with load at least $\bar{l}$. Since $c( l_{ j } + 1 ) = c( \frac{ B }{ l_{ j } + 1 } ) < c( \bar{ l } )$, a migration from $M_k$ to $M_j$ or to a more attractive machine is beneficial for $M_k$. Thus, by choosing the highest suboptimal machine to perform BR, the machines higher than $M_j$ will become balanced, causing $M_j$ to become the highest low machine and a best-response. By Observation \ref{observation:HighStays}, once this happens, $M_j$ remains active in the resulting NE.

We turn to show that $M_j$ must be emptied out if $\frac{n-l_{j}}{m - j} \ge \frac{B}{l_{j}+1}$. Let $M_{ higher } = \{ M_{ j + 1 } , ... , M_{ m } \}$. The average load on machines in $M_{higher}$ is lower than $\frac{B}{l_{j}+1}$. While there are machines lower than $M_j$ ($M_{lower} \neq \emptyset$ ), since the average of loads of $M_{higher}$ is lower than $\frac{B}{l_{j}+1}$, there must be at least one machine in $M_{higher}$ with load $l^{ \prime }$ s.t $ l_j \le l^{ \prime } < \frac{ B }{ l_{ j } + 1 } - 1$ (the subtraction of $1$ is because of the smallest machine). That machine will be a better response than $M_{ j }$ since $c( l_j + 1 ) = c( \frac{ B }{ l_j + 1 } ) > c( l^{ \prime } + 1 )$ or, if $l_j= l^{ \prime }$, since it has a higher index. Afterwards, when $M_j$ is the lowest machine, if there is a higher low machine, it is more attractive than $M_j$, and if there are only high machines, then even if they are balanced, some high machine will have load less than $\frac{ B }{ l_{ j } + 1 }$ and therefore $M_{ j }$ is not a best-response, since $c( l_j ) = c( \frac{ B }{ l_j } ) > c( \frac{ B }{ l_{ j } + 1 } )$. After the machines in $M_{higher}$ are balanced, the jobs on $M_j$ will leave it and it will be emptied.
\end{proof}


We use the above characterization to devise an optimal local deviator rule.

\noindent {\bf An Optimal Deviator Rule:}
We present an optimal local deviator rule. A BR sequence applied with this rule will converge to a NE profile with a maximal possible number of active machines, and will therefore have the optimal social cost.
\begin{algorithm}[h]
\caption{$S_{opt}$ - an optimal local deviator rule}
 \label{alg:coco_opt}
  \begin{algorithmic}[1]
    \STATE Given a profile $p$, let $l_1 \le l_2 \le \cdot \cdot \cdot \le l_{m(p)}$ be its vector profile.
    \STATE {\bf If} $M_{m(p)}$ is a high machine (that is, $l_{ m(p) } \geq l^{*}$) and it is suboptimal {\bf then} return $M_{m(p)}$.
	\STATE {\bf Else if} $M_1$ is suboptimal return $M_{1}$
    \STATE \hspace{0.5in} {\bf Else} $p$ is a NE.
\end{algorithmic}
\end{algorithm}

\begin{theorem}
Algorithm $S_{opt}$ is optimal for every activation cost $B$ and every initial profile $p^0$ of a job scheduling game with conflicting congestion effects.
\end{theorem}
\begin{proof}
We prove that a machine that stays active in an optimal BR sequence will remain active by $S_{opt}$. Let $M_{ i }$ be such a machine. If $L_{ M_ i } \geq l^{*}$ then by Observation \ref{observation:HighStays} it will remain active. Otherwise, when $M_{ i }$ is considered to be chosen by $S_{ opt }$, all the machines with lower loads have been emptied since the rule considers only the lowest machine. Since $M_{ i }$ remains active in an optimal sequence, by Claim \ref{clm:lastingFormula} $\frac{ n - l_{M_i} }{ m - i } > \frac{ B }{ l_{ M_i } + 1 }$ and therefore there is a high machine with load at least $\frac{ B }{ l_{ M_i } + 1 } + 1$ that will be preferred by $S_{ opt }$.
\end{proof}

One may wonder whether deviator rules with poor inefficiency exist in this game.
In Appendix \ref{app:poor} we show an instance for which a natural deviator rule has inefficiency $\Omega( B^{ \frac{ 1 }{ 3 } } )$ and therefore a constant inefficiency is not trivial to ensure.

\section{Conclusions and Open Problems}
\label{sec:conclusion}

A desired property of congestion games is that best-response dynamics always converges to a pure Nash equilibrium.
However, the order in which players are chosen to perform their best-response moves is crucial to the quality of the equilibrium reached.
Unlike previous work, our focus is not on the dynamic's termination or its convergence time, but on the quality of the achieved Nash equilibrium.
Starting from a given profile, two different deviation orders can possibly lead to two different equilibria that differ a lot in their quality.
We introduce a new (worst-case) measure for quantifying the inefficiency of a {\em deviator rule} --- a rule that selects a suboptimal player to perform BR move, and thus, induces the way BRD advance and converge.
In other words, we analyze the ability of a centralized authority that can control the order according to which players perform their best-responses, to lead the players to a good stable profile.

We define the inefficiency of a deviator rule $S$ as the largest ratio, among all initial profiles $p^0$ of a game, between the social cost of the worst NE reachable by $S$ and the social cost of the best NE reachable from $p^0$. We study deviator rules in network formation and job scheduling games.
Our main interest is in deviator rules that are easy to implement and are local in a sense that the information they use about the current profile is limited and based on a restricted set of parameters (referred to as "state vectors"), but we study general deviator rules as well.

We present both positive and negative results for local and global deviator rules in these games. Our positive results show that in some games, even if players act strategically,
controlling the order in which they play can have a significant effect on the quality of the outcome. We show that Min-Path is optimal for symmetric network formation games and present a deviator rule that is optimal for job scheduling game with conflicting congestion effects. Furthermore, we show a dynamic programming based global (non-local) deviator rule that is optimal for network formation games with an underlying SPP network. Our analysis suggest some negative results as well.

Our negative results imply that it would be hard for a central authority to lead the game to a good outcome. We show that in network formation games played on EP networks no local deviator rule can ensure outcome that is better than the PoA, as well as in job scheduling games with identical machines and only negative congestion effect. We also show that it is NP-hard to find optimal outcome in weighted symmetric network formation games, even in the simplest possible network topology (consisting of parallel edges between the source and target nodes).

Additional results shows tight bounds that lie strictly between $1$ and the PoA. In network formation games Min-Path ensures $\mathcal{O}( m )$ inefficiency in single-source instances for SPP networks and $\mathcal{O}( 2^{ m } )$ inefficiency for SPP networks with players corresponding to a proper intervals graph.

Our research can be naturally extended to consider additional state vectors in the games we consider. Furthermore, our measure can be studied in additional games of interest.
%

We suggest a few directions for future research:

First, our results show that taking worst case approach over all initial strategy profiles can lead to negative results that do not distinguish properly between different deviator rules. One way to address this is restricting the initial strategy profile or analysing the inefficiency for different subsets of initial strategy profiles. Such an analysis may suggest a range of deviator rules, each performing well on different initial profiles, and selected to be used based on some of the initial profile's properties.
Another way to address the problem is to consider the average case performance of a deviator rule over initial strategy profiles. In addition to having better results, it may distinguish better between different deviator rules.

Second, we examined only deterministic deviator rules. A {\em mixed deviator rule} returns a distribution over the players rather than choosing one deterministically. The definition of the inefficiency can be naturally extended to inefficiency of a mixed deviator rule using the expected value of the social cost. A \textit{mixed local deviator rule} computes that distribution as a function of local information and satisfies an adjusted IIP condition. One possible definition is the following:
\begin{itemize}
\item Players with the same state vectors must have the same probability to be chosen.
\item Let $p( v )$ denote the \textit{accumulated probability} of a state vector $v$, i.e., let $I$ be the set of players having state vector $v$, then $p( v ) \equiv \sum_{j \in I} p_j$.
A mixed local rule satisfies IIP if in {\em all} profile vectors containing state vectors $v_1, v_2$ it holds that either $p( v_1 ) > p( v_2 )$ or $p( v_1 ) \leq p( v_2 )$. The idea is that a mixed local deviator rule consistently sets a distribution over state vectors rather than players.
\end{itemize}

Third, while we studied only best response dynamics, \textit{better} response dynamics can lead to a wider outcome space. In better response dynamics the central authority has to choose the deviating player along with a new strategy, making sure that the new strategy strictly improves the deviating player's cost. Best response sequence is a special case of better response sequence with possibly a much limited search space. For example, in SPP network formation games presented in subsection \ref{sec:spp}, better response dynamics give the deviator rule the power to solve every segment separately as a pure symmetric game because a deviation of a player in one segment can be done independently and without changing her strategy in all other segments. In potential games, better response dynamics always converges to a Nash equilibrium, but the convergence rate can be worse than in best response dynamics. Consider a symmetric network formation game on a single segment. The Min-Path deviator rule is an optimal local deviator rule for best response dynamics and it must converge within at most $n$ deviations. However, the analysis of better response dynamics is more challenging.
Let us demonstrate this with the following simple example. Consider single segment SPP network with three parallel edges having costs 1, 2 and 3. In the initial profile, these links are used respectively by 0, 2 and 3 players. One can verify that BRD will converge to the middle edge whose cost is 2 - by deviation of the three players on the expensive edge. On the other hand, a better response dynamics may lead to convergence to the cheapest path by first migrating one player from the middle to the expensive edge, and then moving all players (starting from the one on the middle edge) to the cheapest edge.
The social cost is lower but the number of deviations exceeds $n$. Studying the power of better response dynamics as well as the tradeoff between the number of deviations and the quality of the solution is an interesting direction for future work.

Fourth, we presented best response dynamics in which a single player deviates in every iteration. \textit{Coordinated deviator rules} enable a set of players {\em (coalition)} to deviate simultaneously with the restriction that each of the coalition members strictly improves her cost. A Nash equilibrium reached with no beneficial coordinated deviation is denoted a {\em strong Nash equilibrium}. It is well known that not every potential game admits a strong PNE. For example, \cite{HM14} showed that in a network formation game with single-source and multiple-targets played on a network that is not an SPP network, there is no strong PNE. \cite{E+07} showed that in a game on an SPP network there always exist a strong PNE.
Moreover, \cite{FF15} showed that in network formation games, even if a strong PNE exists, not every beneficial coordinated deviations sequence converges to one. They introduced a restricted class of deviations, denoted {\em dominance based beneficial coalition deviation }, which ensures convergence.
There are many open directions in the study of coordinated deviator rules, referring to their inefficiency, as well as to the wide range of possible deviations of the chosen coalition.

\appendix
\section{Network Topologies}
\label{sec:nettopo}
Consider a directed graph $G=(V,E)$ with two designated nodes, a source $s\in V$ and a target $t\in V$.
Assume that every edge is on some $s-t$ path in $G$.
The following operations may be applied on $G$:
\begin{itemize}
\item Identification: the identification of two nodes $v_1,v_2 \in V$ yields a new graph $G'=(V',E')$, where $V'=(V \cup \{v\}) \backslash \{v_1,v_2\}$ and $E'$ includes all the edges of $E$, where each edge that was connected to $v_{1}$ or $v_{2}$ is now connected to $v$ instead. Figuratively, the identification operation is the collapse of two nodes into one.
\item Series composition: Given two networks, $G_1=(V_1,E_1)$ with $s_1,t_1 \in V_1$ and $G_2=(V_2,E_2)$ with $s_2,t_2 \in V_2$, the series composition $G=G_1 \rightarrow G_2$ is the network formed by identifying $t_1$ and $s_2$ in the union network $G'=(V_1 \cup V_2,E_1 \cup E_2)$. In the composed network $G$, the new source is $s_1$ and the new sink is $t_2$.
\item Parallel composition: Given two symmetric networks, $G_1=(V_1,E_1)$ with $s_1,t_1 \in V_1$ and $G_2=(V_2,E_2)$ with $s_2,t_2 \in V_2$, the parallel composition $G=G_1 \parallel G_2$ is the network formed by identifying the nodes $s_1$ and $s_2$ (forming a new source $s$) and the nodes $t_1$ and $t_2$ (forming a new sink $t$) in the union network $G'=(V_1 \cup V_2,E_1 \cup E_2)$.
\end{itemize}

The following classes of network topologies are of special interest:
\begin{itemize}
\item A network $G=(V,E)$ is an extension-parallel (EP) network if one of the following applies:
\begin{itemize}
\item $G$ consists of a single edge.
\item There are two EP networks $G_1,G_2$ such that $G = G_1 \parallel G_2$.
\item There is an EP network $G_1$ and an edge $e$ such that $G = G_1 \rightarrow e$ or $G = e \rightarrow G_1$.
\end{itemize}
\item A network $G=(V,E)$ is a series of parallel-paths (SPP) network if it consists of some edges in parallel composition $G = e_1 \parallel e_2 \parallel ... \parallel e_k$, or if there are two SPP networks $G_1 , G_2$ such that $G = G_1 \rightarrow G_2$.
\end{itemize}

We note that the classes of EP and SPP networks are not comparable, that is, an SPP network is not necessarily an EP one, and vise versa.

\section{Inefficient deviator rule in a job scheduling game with conflicting congestion effects}
\label{app:poor}
\begin{theorem}
The inefficiency of a deviator rule in a job scheduling game with conflicting congestions effects can be $\Omega( B^{ \frac{ 1 }{ 3 } } )$.
\end{theorem}
\begin{proof}
Let $B$ be an activation cost such that $B^{ \frac{ 1 }{ 3 } }$ is an integer. The initial profile consists of $B^{ \frac{ 1 }{ 3 } } + 1$ machines, having initial loads $l_1 = l_2 = ... = l_{ B^{ \frac{ 1 }{ 3 } } }= B^{ \frac{ 1 }{ 3 } }$, and $l_{ B^{ \frac{ 1 }{ 3 } } + 1 } = B$. Notice that the condition specified in Lemma \ref{clm:lastingFormula} holds for every $M_j,~~1 \le j \le B^{ \frac{ 1 }{ 3 } }$. Specifically,
\begin{displaymath}
\frac{ n - l_j }{ m - j } =
\frac{ B + B^{ \frac{ 1 }{ 3 } } ( B^{ \frac{ 1 }{ 3 } } - 1 ) }{ B^{ \frac{ 1 }{ 3 } } -j+1 } >
\frac{ B }{ B^{ \frac{ 1 }{ 3 } } + 1 } =
\frac{ B }{ l_j + 1 }.
\end{displaymath}
In addition, by Observation \ref{observation:HighStays}, the loaded machine will clearly remain active in any NE, since it is high in $p^0$.
We conclude that there exists a BR-sequence in which all $B^{ \frac{ 1 }{ 3 } } + 1$ machines remain active. In particular, a deviator rule that always selects a high machine will lead to a NE in which the jobs are balanced among all $m$ machines. Another possible BR-sequence lets the jobs migrate from $M_1$ to $M_{ B^{ \frac{ 1 }{ 3 } } }$, then from $M_2$, etc. Note that after the first migration $M_{ B^{ \frac{ 1 }{ 3 } } }$ is the only best response for all of the jobs. One can verify that when $M_{ B^{ \frac{ 1 }{ 3 } } - 2 }$ is emptied the jobs on $M_{ B^{ \frac{ 1 }{ 3 } } - 1 }$ can still benefit from migration to $M_{ B^{ \frac{ 1 }{ 3 } } }$ since $c( B^{ \frac{ 1 }{ 3 } } ) = c( \frac{ B }{ B^{ \frac{ 1 }{ 3 } } } ) = c( B^{ \frac{ 2 }{ 3 } } ) > c( B^{ \frac{ 1 }{ 3 } }( B^{ \frac{ 1 }{ 3 } } - 1 ) + 1 )$. Therefore, the only active machines in the NE reached are $M_{ B^{ \frac{ 1 }{ 3 } } }, M_{ B^{ \frac{ 1 }{ 3 } } + 1 }$, i.e., only two active machines. The inefficiency of this sequence is:
\begin{displaymath}
\frac{ c( \frac{ n }{ 2 } ) }
{ c( \frac{ n }{ B^{ \frac{ 1 }{ 3 } } + 1 } ) } =
\frac{ c( \frac{ B + B^{ \frac{ 1 }{ 3 } } \cdot B^{ \frac{ 1 }{ 3 } } }{ 2 } ) }{ c( \frac{ B + B^{ \frac{ 1 }{ 3 } } \cdot B^{ \frac{ 1 }{ 3 } } }{ B^{ \frac{ 1 }{ 3 } } + 1 } ) } =
\frac{ c( B^{ \frac{ 2 }{ 3 } } \frac{ B^{ \frac{ 1 }{ 3 } } + 1 }{ 2 } ) }{ c( B^{ \frac{ 2 }{ 3 } } ) } =
\frac{ B^{ \frac{ 2 }{ 3 } } \frac{ B^{ \frac{ 1 }{ 3 } } + 1 }{ 2 } + \frac{ 2 B^{ \frac{ 1 }{ 3 } } }{ B^{ \frac{ 1 }{ 3 } } + 1 } }{ B^{ \frac{ 2 }{ 3 } } + B^{ \frac{ 1 }{ 3 } } } \geq
\frac{ B^{ \frac{ 2 }{ 3 } } \frac{ B^{ \frac{ 1 }{ 3 } } + 1 }{ 2 } }{ B^{ \frac{ 1 }{ 3 } } ( B^{ \frac{ 1 }{ 3 } } + 1 ) } =
\frac{ B^{ \frac{ 1 }{ 3 } } }{ 2 }
\end{displaymath}
\end{proof}


\begin{thebibliography}{100}

\bibitem{ARV09}
H.~Ackermann, H.~R\"{o}glin, B.~V\"{o}cking. Pure Nash Equilibria in Player-Specific and Weighted
Congestion Games. {\em Theor. Comput. Sci.}, 410(17):1552--1563, 2009.

\bibitem{AD+08}
E.~Anshelevich, A.~Dasgupta, J.~Kleinberg, E.~Tardos, T.~Wexler, T.~Roughgarden.
\newblock The Price of Stability for Network Design with Fair Cost Allocation.
\newblock {\em SIAM J. Comput.}, 38(4):1602--1623, 2008.


\bibitem{AEEMR06}
S.~Albers, S.~Elits, E.~Even-Dar, Y.~Mansour, and L.~Roditty.
\newblock {On Nash Equilibria for a Network Creation Game}.
\newblock In {\em Proc. 17th SODA}, pages 89-98, 2006.

\bibitem{A+10}
N.~Alon, E.D.~Demaine, M.~Hajiaghayi, T.~Leighton.
\newblock Basic Network Creation Games.
\newblock {\em Proc. of the 22nd ACM Symposium on Parallelism in Algorithms and Architectures}, pages 106-–113, 2010.

\bibitem{AKT14}
G.~Avni, O.~Kupferman, and T.~Tamir.
\newblock Network-Formation Games with Regular Objectives.
\newblock {\em J. of Information and Computation}, 251:165--178, 2016.

\bibitem{AT16}
G.~Avni and T.~Tamir, Cost-Sharing Scheduling Games on Restricted Unrelated Machines.
{\em Theoretical Computer Science}, 646:26–-39, 2016.

\bibitem{BF+11}
N.~Berger, M.~Feldman, O.~Neiman, M.~Rosenthal
\newblock Dynamic Inefficiency: Anarchy without Stability
\newblock In {\em Proc. of SAGT}, 2011.

\bibitem{B+15}
V.~Bil\'{o}, M.Flammini, G.~Monaco, L.~Moscardelli
\newblock Computing Approximate Nash Equilibria in Network Congestion Games with Polynomially Decreasing Cost Functions
\newblock {\em Web and Internet Economics}, pages 118--131, 2015.


\bibitem{CF+11}
I.~Caragiannis, A.~Fanelli, N.~Gravin, A.~Skopalik.
\newblock Efficient Computation of Approximate Pure Nash Equilibria in Congestion Games.
\newblock {\em In Proc. of FOCS}, pp. 532–-541, 2011.


\bibitem{CG11}
B.~Chen and S.~G{\"u}rel.
\newblock Efficiency Analysis of Load Balancing Games with and without Activation Costs.
\newblock {\em Journal of Scheduling}, 15(2), pp. 157--164, 2011.

\bibitem{CS11}
S.~Chien and A.~Sinclair.
\newblock Convergence to Approximate Nash Equilibria in Congestion Games.
\newblock {\em Games and Economic Behavior} 71(2): 315--327, 2011.


\bibitem{CR09}
H.~Chen and T.~Roughgarden. Network Design with Weighted Players,
\emph{Theory of Computing Systems}, 45(2), 302--324, 2009.

\bibitem{CV02}
A.~Czumaj and B.~V{\"o}cking.
\newblock Tight Bounds for Worst-Case Equilibria.
\newblock In {\em Proc. of SODA},
  pp. 413--420, 2002.

\bibitem{SB16}
S.~Durand and B.~Gaujal.
\newblock Complexity and Optimality of the Best Response Algorithm in Random Potential Games.
\newblock \newblock In {\em Proc. of SAGT}, 2016.



\bibitem{E+07}
A.~Epstein, M.~Feldman, Y.~Mansour.
\newblock Strong Equilibrium in Cost Sharing Connection Games.  \newblock In {\em Proc. of the 8th ACM conference on Electronic commerce}, 2007.


\bibitem{EA+03}
E.~Even-Dar, A.~Kesselman, Y.~Mansour.
\newblock Convergence Time to Nash Equilibria.
\newblock In {\em Proc of ICALP}, pp. 502--513, 2003.

\bibitem{EM05}
E.~Even-Dar and Y.~Mansour.
\newblock Fast Convergence of Selfish Rerouting.
\newblock In {\em Proc. of SODA}, pp. 772--781, 2005.

\bibitem{FLMPS03}
A.~Fabrikant, A.~Luthra, E.~Maneva, C.~Papadimitriou, and S.~Shenker.
\newblock On a Network Creation Game.
\newblock In {\em Proc. of PODC}, pages 347-351, 2003.

\bibitem{FPT04}
A.~Fabrikant, C.~Papadimitriou, K.~Talwar.
\newblock The Complexity of Pure Nash Equilibria.
\newblock In {\em Proc. of STOC}, pp. 604--612, 2004.

\bibitem{FF15}
M.~Feldman, O.~Friedler.
\newblock Convergence to Strong Equilibrium in Network Design Games.
\newblock {\em ACM SIGMETRICS Performance Evaluation Review}, 43(3): 71--71, 2015.


\bibitem{FG15}
M.~Feldman and O.~Geri.
Do Capacity Constraints Constrain Coalitions?
In {\em Proc. of AAAI}, pp. 879--885, 2015.

\bibitem{FT15}
 M.~Feldman and T.~Tamir.
\newblock Convergence of Best-Response Dynamics in Games with Conflicting Congestion Effects.
{\em Information Processing Letters} 115(2):112--118, 2015.


\bibitem{FT12}
M.~Feldman and T.~Tamir.
Conflicting Congestion Effects in Resource Allocation Games.
{\em Journal of Operation Research}. 60(3):529--540, 2012.

\bibitem{FH79}
G.~Finn and E.~Horowitz.
\newblock {A Linear Time Approximation Algorithm for Multiprocessor Scheduling.}
\newblock In {\em BIT}, 19(3):312--320, 1979.

\bibitem{F10}
D.~Fotakis.
\newblock Congestion Games with Linearly Independent Paths: Convergence Time and Price of Anarchy.
\newblock {\em Theory Comput. Syst.}, 47(1):113--136, 2010.


\bibitem{HK12}
T.~Harks and M.Klimm.
On the Existence of Pure Nash Equilibria in Weighted Congestion Games
{\em Mathematics of Operations Research}, 37(3):419--436, 2012.

\bibitem{HM14}
R.~Holzman, D.~Monderer.
\newblock Strong Equilibrium in Network Congestion Games: Increasing Versus Decreasing Costs.
\newblock {\em International Journal of Game Theory}, pages 65--67, 2014.

\bibitem{IM+05}
S.~Ieong, R.Mcgrew, E.~Nudelman, Y.~Shoham, Q.~Sun,
\newblock Fast and Compact: A Simple Class of Congestion Games.
\newblock In {\em Proc. of AAAI}, pp. 489--494, 2005.

\bibitem{JM+16}
J.~de Jong, M.~Klimm, M.~Uetz.
\newblock Efficiency of Equilibria in Uniform Matroid Congestion Games.
\newblock In {\em Proc. of SAGT}, 2016.

\bibitem{KL13}
B.~Kawald and P.~Lenzner.
\newblock On Dynamics in Selfish Network Creation.
\newblock {\em In Proc. of SPAA}, pp. 83–-92, 2013.


\bibitem{KP99}
E.~Koutsoupias and C.~Papadimitriou.
\newblock Worst-case Equilibria.
\newblock \emph{Computer Science Review}, 3(2): 65-69, 1999.


\bibitem{MS96}
D.~Monderer and L.~S.~Shapley.
\newblock {Potential Games}.
\newblock \emph{Games and Economic Behavior}, 14: 124--143, 1996.

\bibitem{Mil96}
I.~Milchtaich. Congestion Games with Player Specific Payoff Functions.
{\em Games and Economic Behavior}, 13:111-124, 1996.

\bibitem{Pap01}
C.~H.~Papadimitriou.
\newblock Algorithms, Games, and the Internet.
\newblock In {\em Proc.\ 33rd STOC}, pp. 749--753, 2001.

\bibitem{Ros73}
R.W. Rosenthal.
\newblock A Class of Games Possessing Pure-Strategy {Nash} Equilibria.
\newblock {\em International Journal of Game Theory}, 2:65--67, 1973.

\bibitem{Syr10}
V.~Syrgkanis.
\newblock The Complexity of Equilibria in Cost Sharing Games.
\newblock {\em In Proc. 6th WINE}, pp. 366--377. 2010.

\bibitem{TW07}
E.~Tardos and T.~Wexler. Chapter 19: Network Formation Games and the Potential Function Method,
\newblock In \emph{Algorithmic Game Theory},
\newblock Cambridge University Press, 2007.

\bibitem{Voc07}
B.~V{\"o}cking. Chapter 20: Selfish load balancing.
\newblock In \emph{Algorithmic Game Theory},
\newblock Cambridge University Press, 2007.





\end{thebibliography}
\end{document}